\documentclass[10pt,twocolumn,twoside]{IEEEtran}

\newif\ifsubmission
\submissionfalse 

\usepackage[latin1]{inputenc}
\usepackage{amsmath}
\usepackage{amsfonts}
\usepackage{amssymb}
\usepackage{amstext}
\usepackage{amsthm}
\usepackage{graphicx,subfigure}
\usepackage{url}
\usepackage{color}
\usepackage{xcolor}
\usepackage{bm}
\usepackage{dsfont}
\usepackage{soul}

\newtheorem{lemma}{Lemma}

\newcommand{\VV}{\mathcal V}

\colorlet{mygreen}{green!60!gray}

\graphicspath{{../images/}{images/}}

\begin{document}

\title{Secure Detection of Image Manipulation \\ by means of Random Feature Selection}


\author{Zhipeng Chen, Benedetta Tondi, \IEEEmembership{IEEE}, Xiaolong Li, Rongrong Ni, Yao Zhao, Mauro Barni, \IEEEmembership{Fellow, IEEE},
\thanks{Z. Chen is with Institute of Information Science, Beijing Jiaotong University, Beijing, 100044, CHINA. He is with Beijing Key Laboratory of Advanced Information Science and Network Technology, Beijing Jiaotong University, Beijing, 100044, CHINA. He is with Department of Computer Science, Tangshan Normal University, Tangshan, 063000, CHINA, e-mail: \{chzhpeng@hotmail.com\}.
	
	X. Li, R. Ni and Y. Zhao are with Institute of Information Science, Beijing Jiaotong University, Beijing, 100044, CHINA. They are with Beijing Key Laboratory of Advanced Information Science and Network Technology, Beijing Jiaotong University, Beijing, 100044, CHINA, e-mail: \{lixl@bjtu.edu.cn, rrni@bjtu.edu.cn, yzhao@bjtu.edu.cn\}.
	
	M. Barni and B. Tondi are with the Department of Information Engineering and Mathematics, University of Siena, Via Roma 56, 53100 - Siena, ITALY, phone: +39 0577 234850 (int. 1005), e-mail: \{barni@dii.unisi.it, benedettatondi@gmail.com\}.}
}

\markboth{IEEE TRANSACTIONS ON INFORMATION FORENSICS AND SECURITY, ~Vol.~X, No.~X, XXXXXXX~XXXX}%
{Z. Chen, B. Tondi et al.: Secure Detection of Image Manipulation by means of Random Feature Selection}

\maketitle

\begin{abstract}
We address the problem of data-driven image manipulation detection in the presence of an attacker with limited knowledge about the detector. Specifically, we assume that the attacker knows the architecture of the detector, the training data and the class of features $\VV$ the detector can rely on. In order to get an advantage in his race of arms with the attacker, the analyst designs the detector by relying on a subset of features chosen at random in $\VV$. Given its ignorance about the exact feature set, the adversary attacks a version of the detector based on the entire feature set. In this way, the effectiveness of the attack diminishes since there is no guarantee that attacking a detector working in the full feature space will result in a successful attack against the reduced-feature detector. We theoretically prove that, thanks to random feature selection, the security of the detector increases significantly at the expense of a negligible loss of performance in the absence of attacks. We also provide an experimental validation of the proposed procedure by focusing on the detection of two specific kinds of image manipulations, namely adaptive histogram equalization and median filtering. The experiments confirm the gain in security at the expense of a negligible loss of performance in the absence of attacks.
\end{abstract}


\enlargethispage{\baselineskip}

\IEEEpeerreviewmaketitle
\section{Introduction}


Developing secure image forensic tools, capable of granting good performance even in the presence of an adversary aiming at impeding the forensic analysis, turns out to be a difficult task, given the weakness of the traces the forensic analysis relies on \cite{Gloe07}. As a matter of fact, a number of Counter-Forensics (CF) tools have been developed, whose application hinders a correct image forensic analysis \cite{Boh12}.
Early CF techniques were rather simple, as they consisted in the application of some basic processing operators like noise dithering, recompression, resampling or filtering \cite{Kirch08, Cao10, Stam11}. Though often successful, the application of general post-processing operators, sometimes referred to as laundering, does not guarantee that the forensic traces are completely erased and hence does not necessarily result in the failure of the forensic analysis.

When the attacker has enough information about the forensic algorithm, much more effective CF techniques can be devised. By following the taxonomy introduced in \cite{Biggio13}, we say that we are in a Perfect Knowledge (PK) scenario, when the attacker has complete information about the forensic algorithm used by the analyst. In the PK case, very powerful CF techniques can be developed allowing the attacker to prevent a correct analysis by introducing a limited distortion into the attacked image. Generally speaking, the attacker needs {\em only} to solve an optimisation problem looking for the image which is in {\em some sense} closest to the image under attack and for which the output of the forensic analysis is the wrong one. Even if such an optimisation problem may not be always easy to solve, the exact knowledge of the decision function allows the application of powerful techniques either in closed form \cite{Font12, ComeOPT13, Pasq14}, or by relying on gradient-descent mechanisms \cite{Marra15, BarZhip17}.

In many cases, the attacker has only a Limited Knowledge  about the forensic algorithm \cite{Biggio13}. Let us consider, for example, the case of a machine-learning-based detector looking for the traces left within an image by a particular processing algorithm. The attacker may know only the type of detector used by the analyst, e.g. a Support Vector Machine (SVM) or a Neural Network, and the feature space wherein the analysis is carried out, but he may not have access to the training data. In this case, the attacker can build a surrogate version of the detector by using its own training data, and carry out the attack on the surrogate detector, hoping that the attack will also work on the detector used by the analyst \cite{Biggio13, Marra15, BarZhip17}. In other cases, the attacker may know only the feature space used by the detector. In such a situation, he may resort to a so-called universal CF strategy capable of defeating any detector working in the given feature space \cite{BarniMMSEC}. In most cases, the attack works by modifying the attacked-image so that its representation in the feature space is as close as possible to the representation of an image chosen in a dataset of images belonging to the desired class (e.g. non-compressed or pristine images) \cite{Marra15}. In \cite{BarniIJDCF}, for instance, the attack works by bringing the histogram of the attacked image as close as possible to that of an image belonging to a reference dataset of pristine images, by solving an optimal transport problem. In \cite{Bar14IWDW}, a similar strategy is applied in the DCT domain to attack any double JPEG detector relying on the first order statistics of block DCT coefficients.

Several anti-CF techniques have also been developed. The most common approach consists in looking for the traces left by the CF tools, and develop new forensic algorithms explicitly thought to expose images subjected to specific CF techniques. The search for CF traces can be carried out by relying on new features explicitly designed for this target as in \cite{Lai11,Vale13, ZengMedian14, Derosa15}, or, from a more general perspective, by using the same features of the original forensic technique and design an adversary-aware version of the classifier, as in \cite{BarZhip16, BNT17}. In the latter case, it is recommendable to adopt a large feature space allowing enough flexibility to distinguish original and tampered images as well as images processed with the CF operator. If we assume that the attacker knows that the traces left by the CF tools may themselves be subjected to a forensics analysis, we fall in a situation wherein CF and anti-CF techniques are iteratively developed in a never-ending loop, whose final outcome can hardly be foreseen \cite{BarGon13}. Some attempts to cast the above race of arms between the forensic analyst and the attacker by resorting to game theory have been made in \cite{BT13} and \cite{Stamm12}.
In some cases, it is also possible to predict who between the attacker and the analyst is going to win the game according to the distortion that the attacker may introduce to impede the forensic analysis \cite{BT16}.
Yet in other works, the structure of the detector is designed in such a way to make CF harder. In \cite{biggio2010multiple}, the output of multiple classifiers is exploited to design an ensemble classifier exhibiting improved resilience against adversarial attempts to induce a detection error. In \cite{biggio2015one}, the robustness of a two-class classifier and the inherent superior security of one-class classifiers are exploited to design a {\em one and a half class} detector, that is proven to provide an extra degree of robustness under adversarial conditions. Other approaches to improve the security of machine-learning classifiers are described in \cite{Lowd05, Barreno2010, biggio2015bio}, for applications outside the realm of image forensics.
Despite all the above attempts, however, when the attacker knows the feature space used by the analyst, very powerful CF strategies can be designed whose effectiveness is only partially mitigated by the adoption of anti-CF countermeasures.

In order to restore the possibility of a sound forensic analysis in an adversarial setting, and give the analyst an advantage in his race of arms with the attacker, in this work, we propose to randomise the selection of the feature space wherein the analysis is carried out. To be specific, let us assume that to achieve his goal - hereafter deciding between two hypotheses $H_0$ and $H_1$ about the processing history of the inspected image - the analyst may rely on a large set  $\VV$ of, possibly dependent, features. The number of features used for the analysis may be in the order of several hundreds or even thousands; for instance, they may correspond to the SPAM features described in \cite{Pevny10} or the rich feature set introduced in \cite{Frid12rich}. In most cases, the use of all the features in $\VV$ is not necessary and good results can be achieved even by using a small subset of $\VV$. Our proposal to {\em secure} the forensic analysis is to randomise it by choosing a random subset of $\VV$ - call it $\VV_r$ - and let the analysis rely on $\VV_r$ only; in a certain sense, the {\em randomisation} of the feature space can be regarded as a secret key used by the analyst to improve the security of the analysis. Given its ignorance about the exact feature set used by the analyst, a possibility for the attacker is to attack the entire feature set $\VV$. As we will show throughout the paper, with both theoretical and experimental results, not only attacking the entire set $\VV$ increases the complexity of the attack, but it also diminishes its effectiveness, since there is no guarantee that attacking a detector working in the full feature space will also result in a successful attack against a detector working in the reduced set $\VV_r$.

The rest of this paper is organised as follows. In Section \ref{sec.soa}, we revise prior works using randomisation to improve the security of image forensic techniques and more in general that of any detector or classifier. In Section \ref{sec.theory}, we give a rigorous formulation of image manipulation detection via random feature selection, and analyse the theoretical performance of the random detector under simplified, yet reasonable, assumptions. In Section \ref{sec.exa}, we exemplify the general strategy introduced in Section \ref{sec.theory} by developing an SVM detector based on randomised feature selection within a restricted subset of SPAM features, designed to detect two different kinds of image manipulations, namely adaptive histogram equalization and median filtering. In Section \ref{sec.exp}, we analyse the security of the detectors described in Section \ref{sec.exa} against targeted attacks carried out in the feature and the pixel domains. As it will be evident from the experimental analysis, random feature selection considerably increases the strength required for a successful attack at the expense of a negligible performance loss in the absence of attacks. Finally, in Section \ref{sec.con}, we draw our conclusions and highlight some directions for future research.

\section{Related works}
\label{sec.soa}

The use of randomisation to improve the security of a detector or a classifier is not an absolute novelty since it has  already been proposed is several security-oriented applications.
%

Early attempts to use randomisation as a countermeasure against attacks were focusing on probing or oracle attacks, i.e. attacks that repeatedly query the detector in order to get information about it and then use such an information to build an input signal that evades the detection\footnote{Probing attacks are usually carried out when no any a-priori knowledge about the classifier is available to the attacker.} (see for instance \cite{linnartz1998analysis} for an example related to one-bit watermarking and \cite{Barreno2010} for the use of randomisation in the context of machine learning). In all these works, the outcome of the detector is randomised by letting the output be chosen at random for points in the proximity of the detection boundary. In general, boundary randomisation only increases the effort necessary to the attacker to enter (or exit) the detection region; in addition, it also causes a loss of performance in the absence of attacks, that is, the robustness of the system decreases.

Other strategies exploiting randomness to prevent the attacker from gathering information about the detector have been adopted for spam filtering, intrusion detection \cite{biggio2010multiple} and multimedia authentication \cite{koval2008security}. A rather common approach consists in the use of randomisation in conjunction with multiple classifiers. In \cite{breiman1996bagging}, for instance, randomness pertains to the selection of the training samples of the individual classifiers, while in \cite{ho1998random} is associated to the selection of the features used by the classifiers (hereafter referred to as random subspace selection), each of which is trained on the entire training set.
Another randomisation strategy, used in conjunction with multiple classifiers, has been proposed in \cite{biggio2008adversarial} and experimentally evaluated for spam filtering applications. Specifically, an additional source of randomness is introduced in the choice of the
weights of the filtering modules of the individual classifiers.
Random subspace selection has also been adopted in steganalysis \cite{Kod11}, though with a different goal, that is to reduce the problems encountered when working with extremely high-dimensional feature spaces.

The use of randomization for security purposes is also common in multimedia hashing for authentication and content identification \cite{koval2008security, villan2007tamper}. In these works, {\em random projections} are often employed to elude attacks: specifically, a secret key is used to generate a random matrix which is then employed to generate the hash bits of the content, by first projecting the to-be-authenticated signal on the directions identified by the rows of the matrix and then comparing the absolute value of the projections against a threshold. Despite the apparent similarity, the use of random projections for multimedia hashing differs substantially from the technique proposed in this paper. In multimedia hashing applications, the random projection is applied directly in the pixel or in a transformed domain, and is possibly followed by the use of channel coding to improve robustness against noise. The kind of traces we are looking for in multimedia forensics applications, however, are so weak that a completely random choice of the feature space would not work. For this reason, in our system, randomisation is applied within a set of features explicitly designed to work in a specific multimedia forensics scenario.
As a matter of fact, the system proposed in this paper can be seen as the application of the random projection method directly in the feature space, with the projection matrix designed in such way to contain exactly one non-zero in each row, with the additional computational advantage that only the selected features need to be calculated, while a full projection matrix would require the computation of the entire feature set.

The use of feature selection for security purposes has also been proposed in \cite{zhang2016adversarial}. Although the idea of resorting to a reduced feature set is similar to our proposal, the set up considered in \cite{zhang2016adversarial} differs considerably from the one adopted in this paper. In \cite{zhang2016adversarial}, in fact, the authors search for the best reduced feature set against an attacker with perfect knowledge about the detector, i.e. an attacker who knows the choice of the feature subset made by the defender. This is different from the scenario considered here, where feature randomization works as a kind of secret key.

\section{Secure detection by random feature selection}
\label{sec.theory}

In this section, we first describe the security assumptions behind our work, then we give a rigorous definition of binary detection based on Random Feature Selection (RFS) and provide a theoretical analysis to evaluate the security vs robustness trade-off under a simple statistical model. Though derived under simplified assumptions, the theoretical analysis is an insightful one since it provides useful insights on the impact that the statistics of the host features has on the security of the randomised detector. In addition, it permits to analyse the dependence of classification accuracy on the number of selected features both in the presence and in the absence of attacks. Even if the paper focuses on RFS, the theoretical framework is a general one and can also be used to analyse other kinds of (linear) feature randomisation.

\subsection{Security model}

The security model adopted in this paper is depicted in Fig. \ref{fig:secmodel}. As shown in the picture, we assume that the attack is carried out directly in the feature space. This is not a problem when an invertible relationship exists between the pixel and the feature domain, as it is the case, for instance, of detectors based on block DCT coefficients or their histogram \cite{BarniMMSEC, Bar14IWDW}. In other cases, however, mapping back the attack into the pixel domain may not be easy, or even possible, since there is no guarantee that the attacked feature vector is a feasible one, i.e., that an image exists whose features are equal to those resulting from the attack. Nevertheless, analysing the performance of randomized feature detection in this scenario, which is undoubtedly more favourable for the attacker, provides interesting insights about the security improvement achieved by RFS.
We also assume that the attacker is interested in inducing a missed detection error, while minimising the distortion introduced in the image as a consequence of the attack.

\begin{figure}[t!]
	\centering	
		\includegraphics[width=0.99\columnwidth]{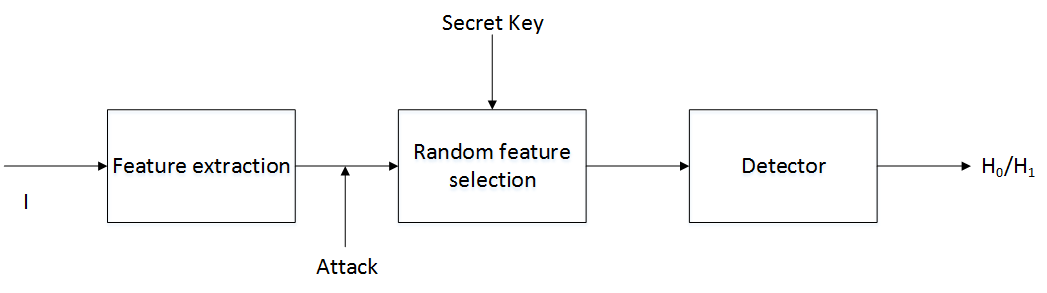}	
	\caption{Security model for the RFS detector.}
	\label{fig:secmodel}
\end{figure}

With regard to the knowledge available to the attacker, we assume that he knows everything with the exception of the subset of features used by the detector. This includes knowledge of the detector architecture, the full feature set $\VV$, and the training set (in case of a detector based on machine learning). We also assume that it is not possible for the attacker to query the detector to observe the result of the detection on one or more specific inputs. This means that the attacker can not run oracle or mimicry attacks on the reduced feature detector to infer information about the subset of features used by the detector \cite{cox1997public,wang2006anagram}.

In order to assess the security of the RFS approach when faced with attacks carried  out in the pixel domain, in section \ref{sec.pixdomain} we show the results of some experiments where an RFS detector trained to detect adaptive histogram equalization and median filtering is subject to an attack working in the pixel domain (see section \ref{subsec.att2} for a description of the attack).

\subsection{Problem formulation}

Let ${\bf v}$ be an $ n $-long column vector with the image features the detector relies on. The detector aims at distinguishing between two hypotheses: $ H_0 $ - "the image is manipulated", and $ H_1 $ - "the image is original". We assume that the probability density function of ${\bf v}$ under the two hypotheses is as follows:
\begin{align}
& H_0: {\bf v} \simeq {\mathcal N}({\bf u}, \Sigma_0) \\ \nonumber
& H_1: {\bf v} \simeq {\mathcal N}({\bf -u}, \Sigma_1),
\label{eq.model}
\end{align}
where ${\mathcal N}({\bf u}, \Sigma_0)$ (res., ${\mathcal N}(-{\bf u}, \Sigma_1)$), is a multivariate Gaussian distribution with mean $\bf u$  (res., $-\bf u$), and covariance matrix $\Sigma_0$ (res., $\Sigma_1$). Note that assuming that the mean vectors under the two hypotheses are one the opposite of the other does not cause any loss of generality. In fact, if this is not the case, we can always apply a translation of the feature vector, for which such an assumption holds.

The derivation of the optimal detector, under the assumption that the a-priori probabilities of $H_0$ and $H_1$ are equal, passes through the computation of the log-likelihood probabilities of observing ${\bf v}$ under $H_0$ and $H_1$. 
In the following, we will carry out our analysis by assuming that $\Sigma_0 = \Sigma_1= \Sigma$. In this case, the optimum detector is a correlation-based detector \cite{Kay}, according to which the detector decides for $H_0$ if:
\begin{equation}
\label{eq.optdet}
\rho = (\Sigma^{-1} {\bf u})^T {\bf v} > 0.
\end{equation}
In the rest of the section, we will refer to the detector defined by equation \eqref{eq.optdet} as the optimum {\em full-feature} detector.

In order to improve the security of the system, the detector randomises the decision strategy by relying on a reduced feature vector ${\bf v}_r = S {\bf v}$ where $S$ is a random $k\times n$-dimensional matrix. Several different randomisation strategies can be adopted according to the form of $S$. For the RFS detector proposed in this paper, $S$ is a matrix whose entries are all zeros except for a single element of each row which is equal to one. In addition, all nonzero entries are located in different columns. This corresponds to form ${\bf v}_r$ by selecting at random $k$ elements of ${\bf v}$. Another possibility consists in generating all the elements of $S$ independently according to a given distribution (e.g. Gaussian) and normalising the entries so that the Euclidean norm of each row is equal to one. In this way, the randomised detector relies on the projections of the vector ${\bf v}$ on $k$ random directions (Random Projection - RP).



\subsection{Theoretical analysis}

In this section, we analyse the trade-off between security and robustness by evaluating the performance of the randomised detector with and without attacks. As we will see, by lowering $k$, the security of the detector increases at the price of a loss of performance in the absence of attacks, with a better trade-off reached by the RFS detector.

\subsubsection{Performance in the absence of attacks (robustness)}

As highlighted in the previous section, the sufficient statistics for the optimum full-feature detector is given by $\rho  = \left( {{{\Sigma}^{ - 1}}{\bf u} } \right)^T {\bf v}$.
The performance of the full-feature detector depend on the statistics of $\rho$. Due to the normality of ${\bf v}$, $\rho$ is also normally distributed with mean and variance under $H_0$ given by:
\begin{equation}
\label{eq.stat_rho}
\begin{split}
E\left[ \rho | H_0 \right] &= \left( {{{\Sigma}^{ - 1}}{\bf u} } \right)^T {\bf u}  = {\bf u} ^T {{\Sigma}^{ - 1}}{\bf u} \\
var[\rho | H_0 ] &= ({{\Sigma}^{ - 1}}{\bf u} )^T {\Sigma}{\Sigma}^{ - 1}{\bf u}  = {\bf u} ^T{\Sigma}^{ - 1}{\bf u}.
\end{split}
\end{equation}
Similar values are obtained under $H_1$, by replacing ${\bf u}$ with $-{\bf u}$. The error probability of the detector is related to the $ z $-value of the normal distribution, which is equal to:
\begin{align}
z = \frac{{{{\bf u} ^T}{{\Sigma}^{ - 1}}{\bf u} }}{{{{\left( {{{\bf u} ^T}{{\Sigma}^{ - 1}}{\bf u} } \right)}^{1/2}}}} = {\left( {{{\bf u} ^T}{{\Sigma}^{ - 1}}{\bf u} } \right)^{1/2}}.
\label{eq.z2_fullD}
\end{align}
Note that $z$ is always positive since $\Sigma$ is a positive-definite matrix and that higher values of $z$ correspond to a lower error probability. 
Due to the symmetry of the problem, the two error probabilities, i.e., the probability of deciding in favour of $H_0$ when $H_1$ holds and the reverse probability of deciding for $H_1$ when $H_0$ holds, have the same value; hence, in the following, we will generally refer to the error probability as $P_e$.

In the case of randomised detection, the feature vector is: $ {\bf v}_r=S \cdot {\bf v} $. For a given $S$, the statistics of the observations under the two hypotheses are as follows:
\begin{equation}
\begin{split}
{H_0} & :{{\bf v}_r} \simeq \;N({\bf u}_r, \Sigma_r)\\
{H_1} & :{{\bf v}_r}\; \simeq N( - {\bf u}_r, \Sigma_r),
\end{split}
\label{eq.model2}
\end{equation}
where we let ${\bf u}_r = S {\bf u}$ and $\Sigma_r = S{\Sigma}{S^T}$. The optimum detector now decides for $ H_0 $ when ${\rho _r} = \left( {\Sigma_r^{- 1}{\bf u} _r} \right)^T {\bf v}_r > 0$.
As for the full detector, $\rho _r$ is a Gaussian r.v. with mean and variance (under $H_0$) given by:
\begin{equation}
\begin{split}
E\left[ \rho_r | H_0 \right] &= {\bf u}_r^T {\Sigma}_r^{-1}{\bf u}_r \\
var[\rho_r | H_0 ] &= {\bf u}_r^T{\Sigma}_r^{-1}{\bf u}_r.
\end{split}
\label{eq.rho2_reduced_stat}
\end{equation}
Once again the error probability depends on the $z$-value of the randomized detector, which is  ${z_r} = \left( {\bf u}_r^T \Sigma_r^{-1} {\bf u}_r \right)^{1/2}$.
By introducing the factor
\begin{equation}
\label{eq.eta}
    \eta = \frac{{\bf u} _r^T \Sigma_r^{-1} {\bf u}_r}{{\bf u}^T \Sigma^{-1} {\bf u}},
\end{equation}
the $z$-value of the randomised detector can be related to that of the full-feature detector as follows:
%
\begin{equation}
\label{eq.2z}
    z_r = \sqrt{\eta}z.
\end{equation}
Given that $\eta$ is always lower than one and decreases when $k$ decreases, equations \eqref{eq.eta} and \eqref{eq.2z} determine the loss of performance due to the use of the randomised reduced-feature detector instead of the full-feature one. Since in this case the errors are due to the natural variability of the observed features, the loss of performance can be regarded as a loss of robustness.

\subsubsection{Performance under attack (security)}

Given that the attacker does not know the subset of features used by the randomised detector and, in principle, he does not even know about the randomization-based defence mechanism, we can assume that he keeps attacking the full-feature detector (an alternative strategy is considered in Section \ref{sec.attackNew}). Without loss of generality, we will assume that the attacker takes a sequence generated under $H_0$ and modifies it in such a way that the detector decides in favour of $H_1$.
The optimum attack is the one that succeeds in inducing a decision error while minimising the distortion introduced within ${\bf v}$. Such an attack is obtained by moving the vector ${\bf v}$ orthogonally to the decision boundary until $\rho = 0$, leading to:
\begin{equation}
    {\bf v}^* = {\bf v} - \alpha \frac{\left( {{{\Sigma}^{ - 1}}{\bf u} } \right)^T {\bf v}}{\| \Sigma^{-1} {\bf u} \|^2 } \cdot \Sigma^{-1} {\bf u},
\label{eq.optattack}
\end{equation}
where $\bf{v}^*$ is the attacked feature vector and $\alpha$ is a parameter controlling the strength of the attack: with $\alpha = 1$ the attacked vector is moved exactly on the decision boundary, however the attacker may decide to use a larger $\alpha$ (introducing a larger distortion) so to move the attacked vector more inside the wrong decision region, hence increasing the probability  that the attack is successful also against the randomised detector.

By construction, when applied against the full feature detector the above attack is always successful. We now investigate the effect of the attack defined in \eqref{eq.optattack}, when the analyst uses a randomised detector based on a reduced set of features. We start by observing that, as a consequence of the attack, the value of $\rho_r$ evaluated by the randomised detector is:
\begin{equation}
\label{eq.attacked_rho}
\rho_r = {\bf u}_r^T \Sigma_r^{-1} {\bf v}_r - \alpha \frac{{\bf u}_r^T \Sigma_r^{-1} {\bf w}_r}{||{\bf w} ||^2} {\bf w}^T {\bf v},
\end{equation}
where we let ${\bf w} = \Sigma^{-1} {\bf u}$, ${\bf w}_r = S {\bf w}$, and $\theta = \alpha ({\bf u}_r^T \Sigma_r^{-1} {\bf w}_r/||{\bf w} ||^2)$.
The statistics of $\rho_r$ are given by the following lemma.
\begin{lemma}
\label{eq.lemmastat}
By letting $y = {\bf u}_r^T \Sigma_r^{-1} {\bf u}_r$, and $x = {\bf u}^T \Sigma^{-1} {\bf u}$, we have that under $H_0$:
\begin{equation}
\begin{aligned}
\label{eq.mean_var_rhor}
    E[\rho_r] = y - \theta x, \hspace{0.5cm}  var[\rho_r] = y + \theta^2 x - 2 \theta y.
\end{aligned}    
\end{equation}
%
%
\end{lemma}
\begin{proof}
See the appendix.
\end{proof}

It is worth observing that Lemma \ref{eq.lemmastat} holds under the assumption that the attacker applies equation \eqref{eq.optattack} even when the full-feature detector decides for  $H_1$, that is when  $\rho < 0$. In general this will not be the case, since when the detector already  makes an error due to the presence of noise, the attacker has no interest to modify ${\bf v}$. In fact, in this case, the result of the application of equation \eqref{eq.optattack} would be the correction of the error made by the detector. Given that deriving the statistics of $\rho_r$ by taking into account that the attack is present only when $\rho > 0$ is an intractable problem, we assume that equation \eqref{eq.optattack} always holds. The results obtained in this way provide a good approximation of the real performance of the system when the error probability of the full-feature detector in the absence of attacks is negligible, i.e., when $z$ is much larger than 1.

%
%
The $z$-value of the randomised detector under attack  can then be written as
\begin{equation}
\label{eq.zatt2}
    z_{att} = z \frac{\eta - \theta}{\sqrt{\eta + \theta^2 -2 \eta \theta}},
\end{equation}
where $\eta$, defined as in \eqref{eq.eta}, is equal to $y/x$. Note that, due to the attack, $z_{att}$ can also take negative values thus resulting in a large error probability.

The interpretation of equation \eqref{eq.zatt2} is rather difficult since $\eta$ and $\theta$ depend on ${\bf v}$ in a complicated way. In the next section we will use numerical simulations to get more insight into the performance predicted by \eqref{eq.zatt2}. Here we observe that the expression of $z_{att}$ can be simplified considerably when $\Sigma = \sigma^2 I$, that is when the features are independent and have all the same variance, and the rows of $S$ are orthogonal, as with RFS. In this case, it is easy to see that:
\begin{equation}
\label{eq.iid}
\begin{aligned}
    \theta = \alpha \frac{|| {\bf u}_r||^2}{|| {\bf u}||^2}, ~\eta &= \frac{|| {\bf u}_r||^2}{|| {\bf u}||^2},
\end{aligned}
\end{equation}
and hence $\theta = \alpha \eta$. Equation \eqref{eq.zatt2} can then be rewritten as:
\begin{equation}
\label{eq.zatt2iid}
    z_{att} = z \frac{\eta (1- \alpha)}{\sqrt{\eta + \alpha^2 \eta^2 - 2 \alpha \eta^2}}.
\end{equation}
From equations \eqref{eq.iid} and \eqref{eq.zatt2iid}, we see that the error probability under attack depends only on $\eta$ , that is the ratio of the norm of the vector with the mean value of the reduced set of features and the norm of the full-feature mean vector. Clearly, when the number $k$ of features used by the randomised detector decreases, the value of $\eta$ decreases as well. Given that the numerator in \eqref{eq.zatt2iid} is either null or negative, and given that the quantity
\begin{equation}
\label{eq.deponeta}
    \frac{\eta}{\sqrt{\eta + \alpha^2 \eta^2 - 2 \alpha \eta^2}} = \frac{1}{\sqrt{\frac{1}{\eta} + \alpha^2 - 2 \alpha}}
\end{equation}
increases when $\eta$ decreases, the error probability under attack decreases with $k$. In fact, if $\eta$ approaches zero, i.e., $k = 1$, $z_{att}$ will be close to 0, and the error probability under attack tends to 0.5. In other
words, the probability that an attack against the full feature detector is also affective against a reduced detector based on one feature only is 0.5 (the improved security comes at the price of a reduced effectiveness in the
absence of attacks, as stated by equation \eqref{eq.2z}). As we will see in the next section, even better results are obtained when the features are not independent. Expectedly, it is also easy to see that when $\eta = 1$, that is
when all the features are used, $z_{att} = -z$, hence resulting in a very large error probability\footnote{In fact, when $\eta = 1$ the error probability should be equal to 1. This is not the case in the present analysis due to the
assumption we made that the attacker always applies equation \eqref{eq.optattack}, even when $\rho < 0$.}.

\subsection{Numerical results}
\label{subsec.numres}

In this section, we use numerical analysis to study the performance of the randomised detector both in the presence and in the absence of attacks as predicted by equations \eqref{eq.2z} and \eqref{eq.zatt2}.

We start with the simple case of i.i.d. features. The performance predicted by the theory are reported in Fig. \ref{fig:iid}, where we show the dependence on $k$ of the missed detection error probability both in the presence (upper curves) and in the absence (lower curves) of attacks\footnote{Due to the symmetry of the problem, in the absence of attacks, the missed detection probability is equal to the overall error probability.}. The plots have been obtained by letting $n = 300$, $z = 4$, and averaging over 500 random choices of the matrix $S$. We set $\alpha = 1.2$ for the leftmost plot and $\alpha = 2$ for the rightmost.
As it can be seen, no particular difference can be noticed between the RFS and RP detectors. The security of the randomised detector increases for lower values of $k$,
while the performance in the absence attacks decreases. As predicted by equation \eqref{eq.deponeta}, when the number of features used by the detector tends to 1, the
error probability in the presence of attacks tends to 0.5, thus showing that the attack designed to defeat the full-feature detector fails almost half of the times when the reduced feature detector is used. Expectedly, the attack is more successful for larger $\alpha$.


\begin{figure}[t]
	\centering	
	\subfigure[]{
		\includegraphics[width=0.48\columnwidth, trim={40 180 50 180},clip]{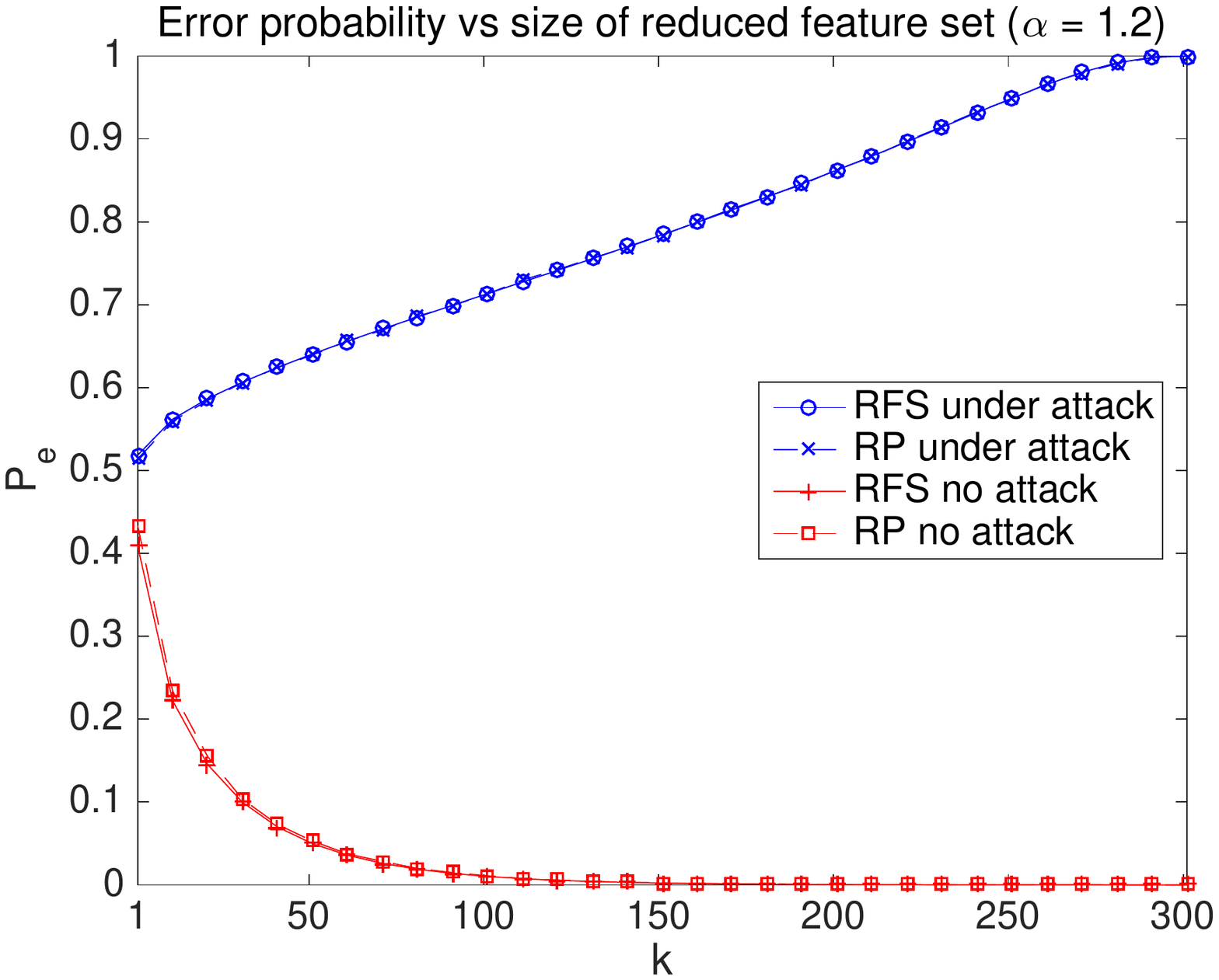}}
	\subfigure[]{
		\includegraphics[width=0.48\columnwidth, trim={40 180 50 180},clip]{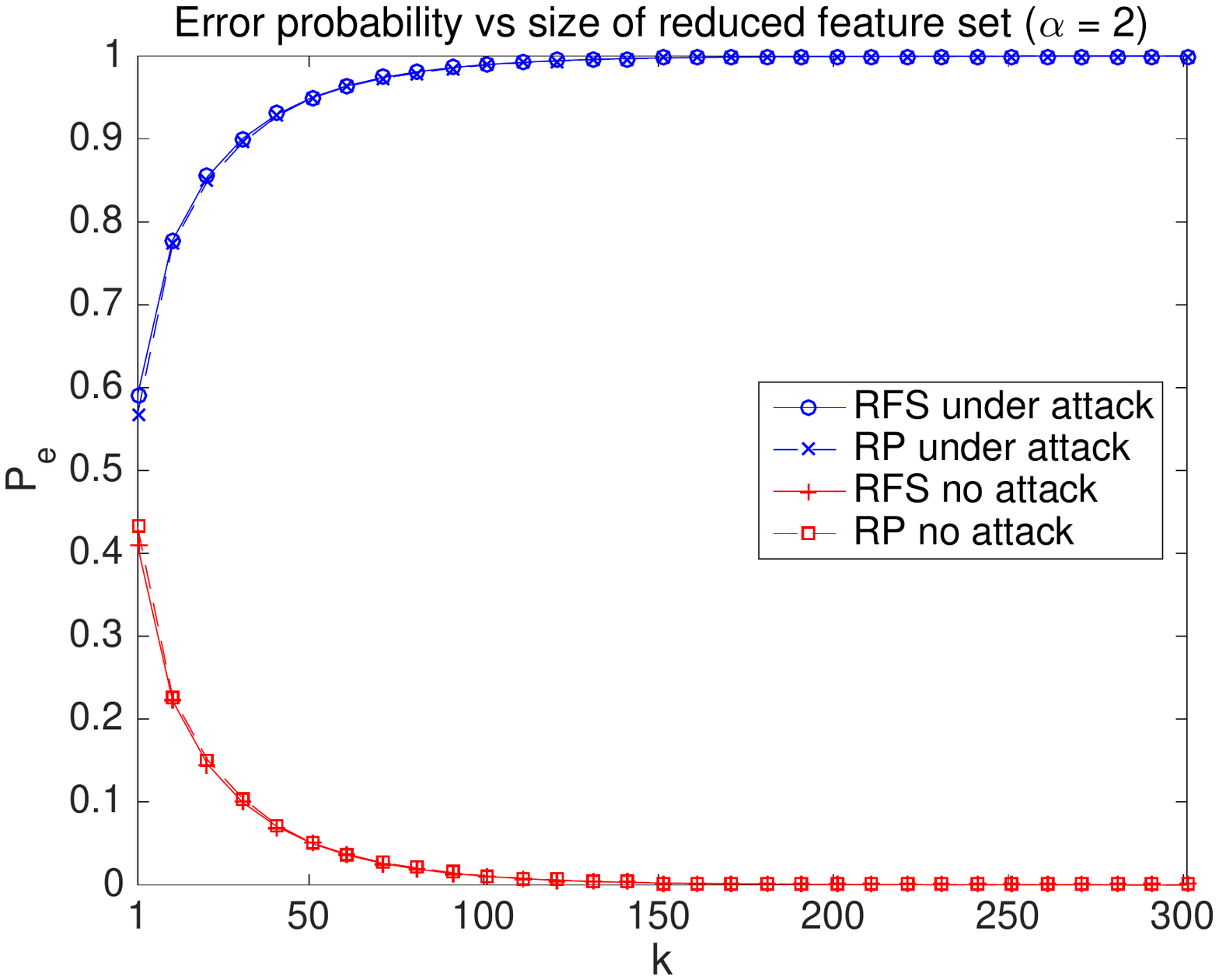}}	
	\caption{Missed detection error probability of the randomised feature detector with and without attacks in the case of i.i.d. features. The plots have been obtained by letting $z = 4$, $n = 300$, $\alpha = 1.2$ (a), $\alpha = 2$ (b), and averaging over 200 random choices of the matrix $S$. }
	\label{fig:iid} 
\end{figure}

\begin{figure}[t]
	\centering	
	\subfigure[]{
		\includegraphics[width=0.48\columnwidth, trim={30 180 50 180},clip]{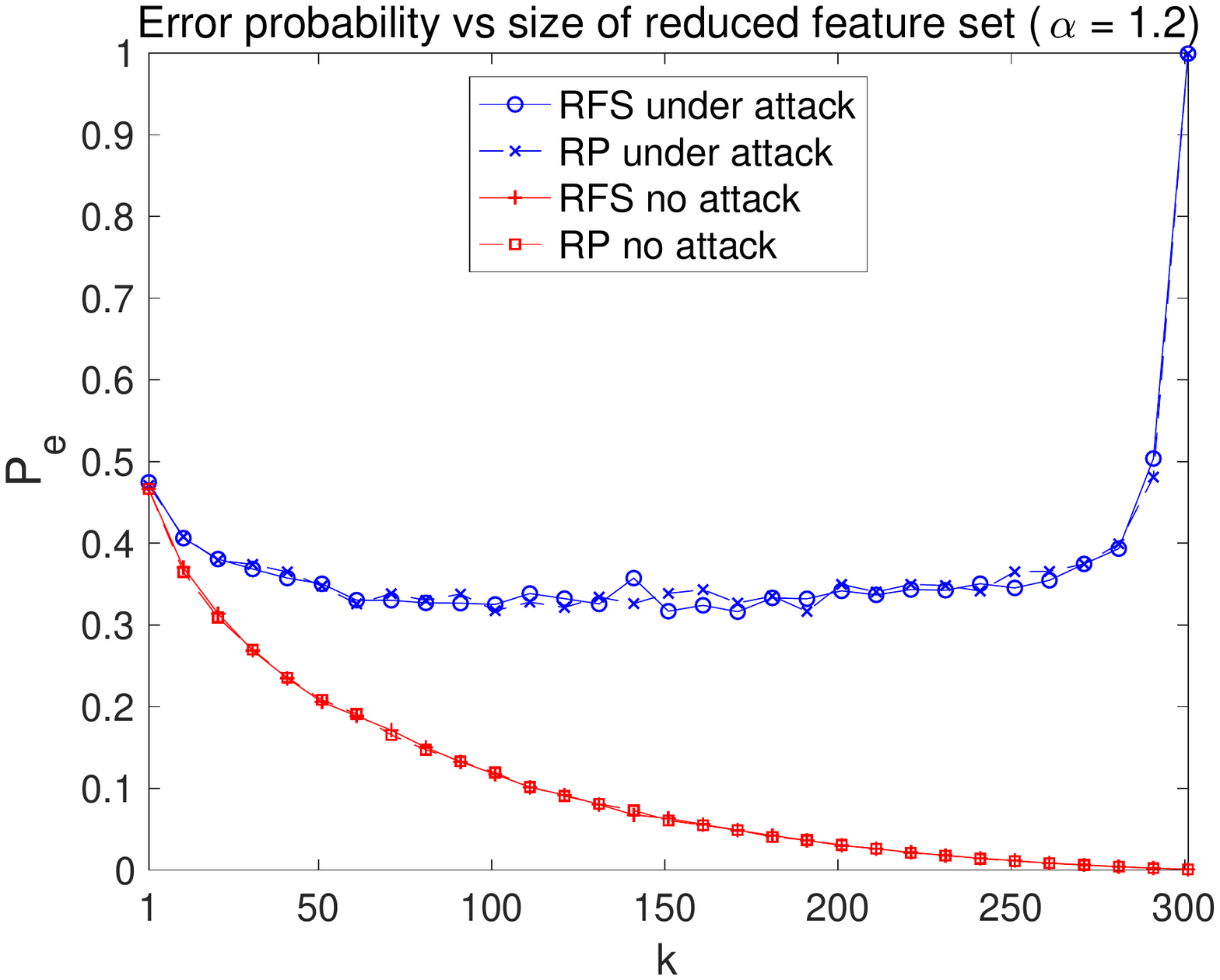}}
	\subfigure[]{
		\includegraphics[width=0.48\columnwidth, trim={30 180 50 180},clip]{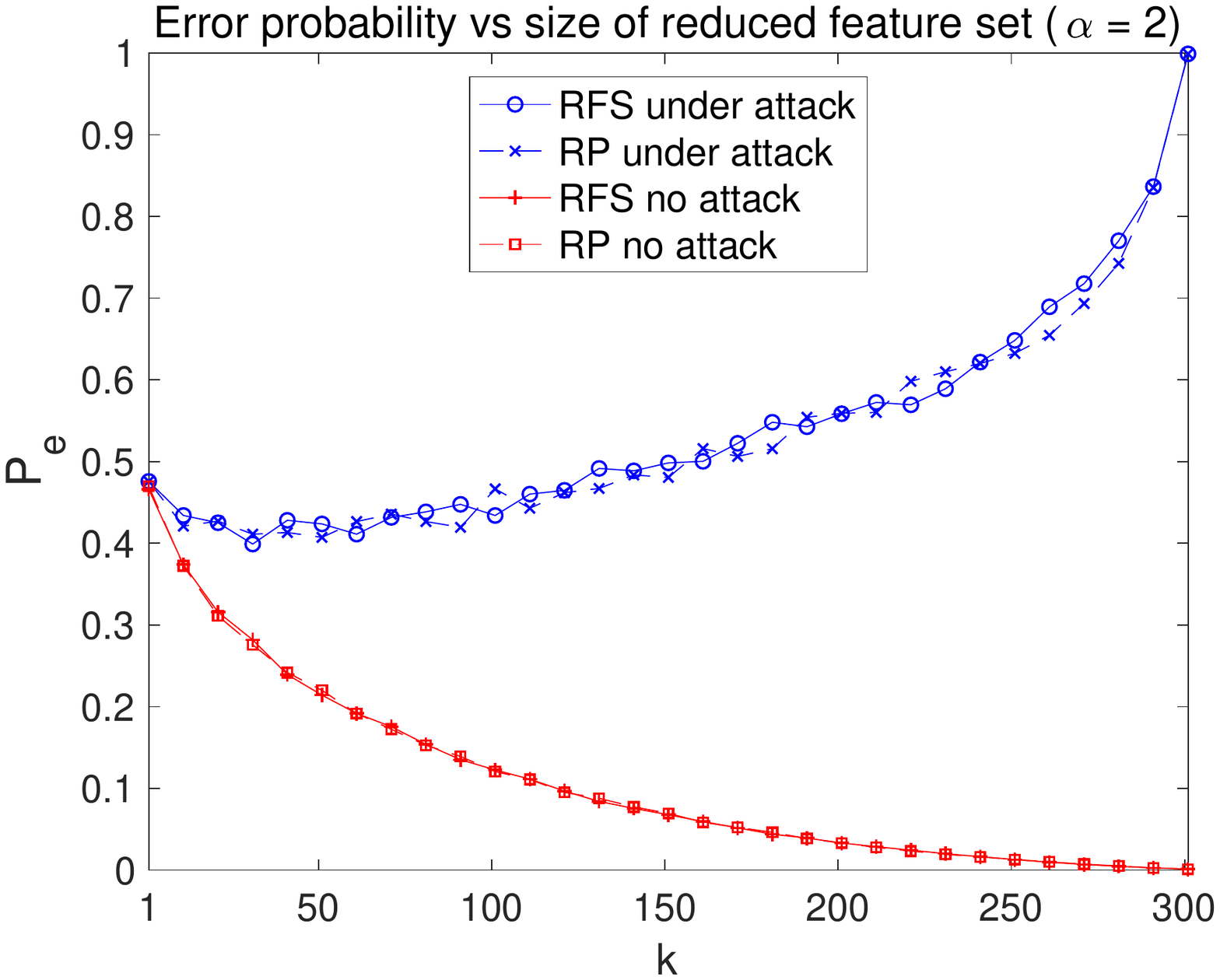}}
			\subfigure[]{
		\includegraphics[width=0.48\columnwidth, trim={40 180 50 180},clip]{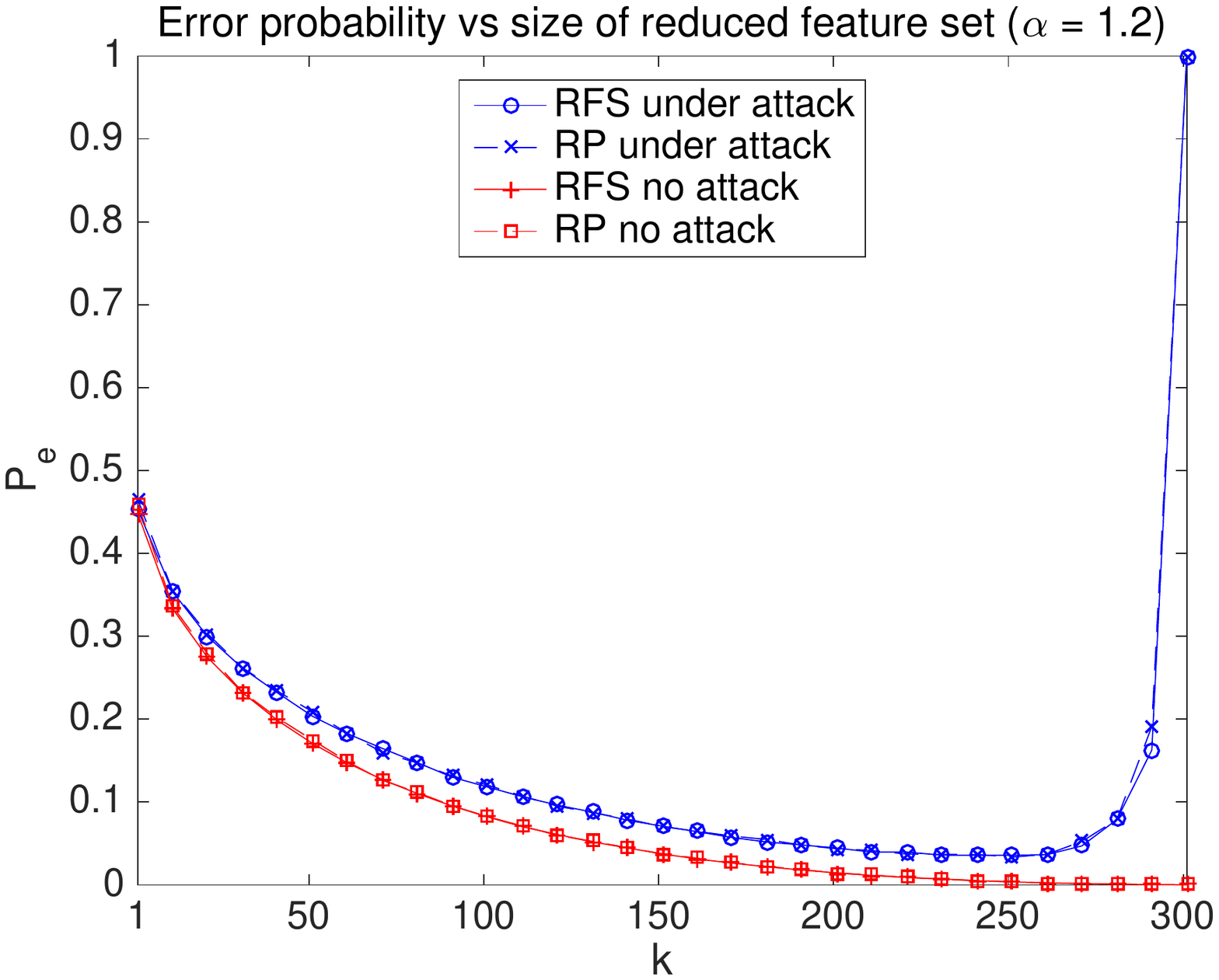}}
	\subfigure[]{
		\includegraphics[width=0.48\columnwidth, trim={40 180 50 180},clip]{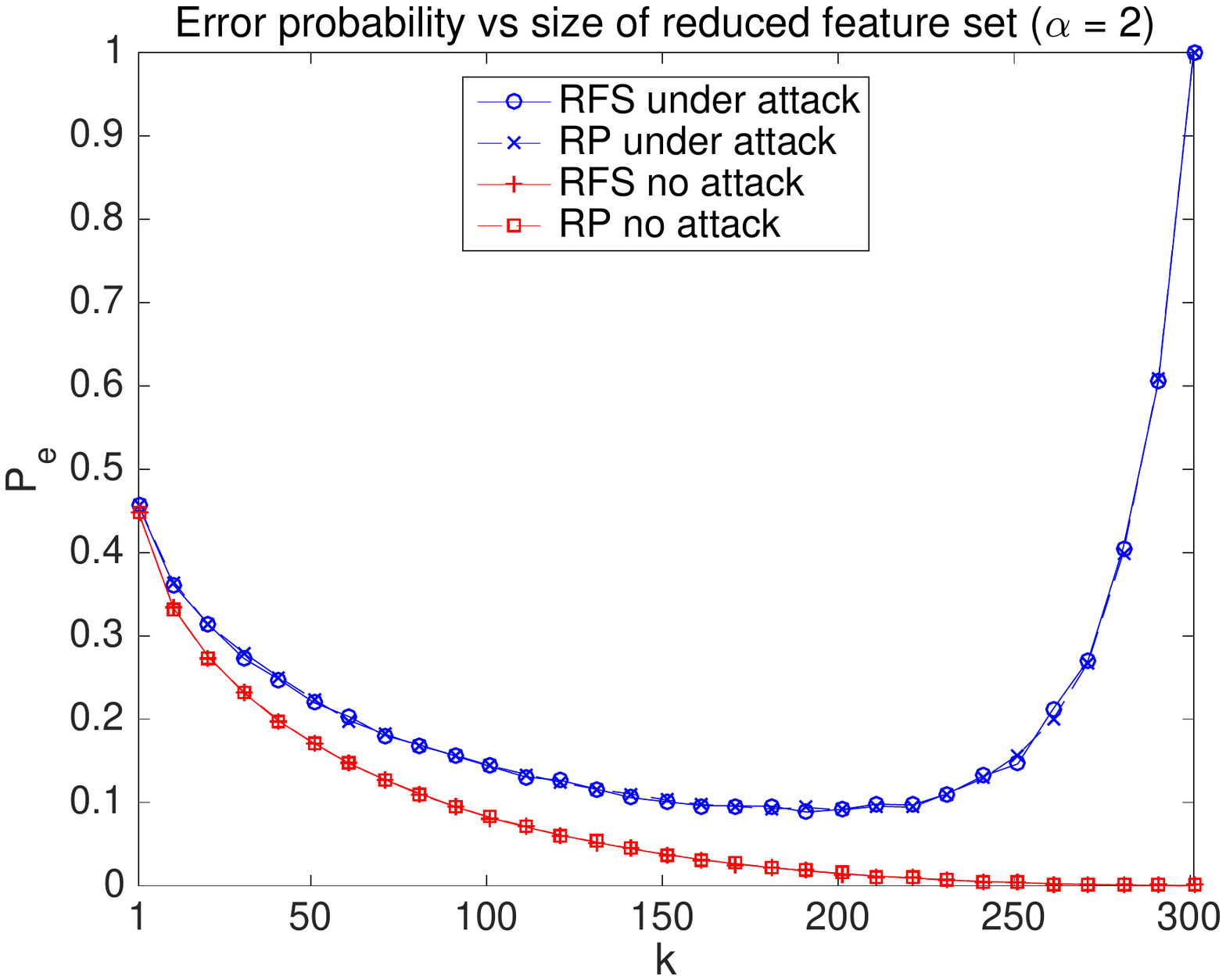}}	
	\caption{Missed detection error probability of the randomised feature detector with and without attacks, in the case of dependent features. The plots have been obtained by letting $n = 300$, $\alpha = 1.2$ (a)-(c), and $\alpha = 2$, (b)-(d). In (a) and (b) the features variance is not normalised while  (c) and (d) refer to the case of normalised features. In all cases the error probability has been obtained by averaging over 500 random choices of the host feature statistics and the matrix $S$. The average value of $z$ was: (a) $z = 4.9$, (b) $z =  5$, (c)  $z = 4.8$, (d) $z = 4.9$.}
\label{fig:dep} 
\end{figure}

We now consider the more general case of dependent features. To do so, we set the statistics of the host features as follows:
i) constant mean vector ${\bf u}$ (we also run some simulations with a randomly generated mean vector obtaining very similar results), ii) random covariance matrix
constructed by first generating a diagonal matrix with uniformly distributed random diagonal entries, and then randomly rotating the diagonal matrix so to obtain dependent
features. We observe that in this way the features have different variances, however, especially after the random rotation, the difference
among the variances is not big. This agrees with a practical setup in which the detector relies on normalised features. To force a different variance among features, we considered an additional setup in which the variance of the features is scaled after the random rotation of the covariance matrix. We did so by randomly generating a vector of scaling factors uniformly distributed between 0 and 1 and applying the square root of the scaling factors to the features. We used the square root of the random scaling factors to limit the scale differences among the features. In the following we refer to the first case as dependent normalised features and to the second as dependent features.

Alike in the i.i.d. case, we let $n= 300$ and averaged the results over 500 repetitions, each
time by randomly generating a new covariance matrix and a new matrix $S$. Even in this case, we considered two different values of $\alpha$, namely: $\alpha = 1.2$ and $\alpha = 2$. Fig.
\ref{fig:dep} reports the results that we have obtained. Due to the randomness involved in the generation of the feature statistics, we could not control the exact value of $z$,
the values resulting from each experiment are reported in the caption of the figure. As in Fig. \ref{fig:iid}, no particular difference is observed between the RFS and RP cases,
however, the overall behaviour of the detector is completely different with respect to the i.i.d. case. The error probability under attack now decreases more rapidly when the number of features used by the detector is reduced, thus indicating a high security level. After a certain point, however, the error probability increases again approaching 0.5
when $k$ tends to 1. Such a behaviour can be interpreted as a loss of robustness rather than a loss of security. In fact, the error probability in the absence of attacks exhibits
a similar increase when $k$ decreases, indicating that the detector is not able to distinguish between $H_0$ and $H_1$ by relying on few features only. Of course, such a
problem has an impact also on the error probability in the presence of attacks. Overall, the reduced detector performs better in the case of dependent features, however, higher values of $k$ must be used with respect to the i.i.d. case.
We also observe that a significantly higher security level is obtained for the case of normalised features (cases (c) and (d) in Fig. \ref{fig:dep}). In addition, increasing the value of $\alpha$ does not have a great impact on the success rate of the attack, especially in the case of normalized features.

In order to explain the plots reported in Fig. \ref{fig:dep}, noticeably their difference with respect to Fig. \ref{fig:iid}, we need to analyse in more detail the reason behind the security improvement achieved through random feature selection. We will do so by referring to the RFS case, however the same considerations can be applied to the RP case. We start by observing that the optimal attack given in the rightmost part of equation \eqref{eq.optattack} can be decomposed in two parts: a scaling factor $a$ and a direction ${\bf e}_{att}$:
\begin{equation}
\label{eq.attack_decomp}
\begin{aligned}
& a = \alpha \frac{\left( {{{\Sigma}^{ - 1}}{\bf u} } \right)^T {\bf v}}{\| \Sigma^{-1} {\bf u} \|}\\
& {\bf e}_{att} = \frac{\Sigma^{-1} {\bf u}}{\| \Sigma^{-1} {\bf u} \|}.
\end{aligned}
\end{equation}

The direction ${\bf e}_{att}$ ensures the optimality of the attack, since it is chosen so to be orthogonal to the boundary of the decision regions. The scaling factor $a$ is chosen  in such a way to move the attacked feature vector exactly on the boundary of the decision regions ($\alpha =1$) or inside the target region ($\alpha > 1$). With the randomized feature detector, the attack vector is projected onto a subspace with lower dimensionality and both the scaling factor and the direction of the projected attack vector are no more optimal. The non-optimality of the scaling factor ({\em scale mismatch}) implies that sometimes the magnitude of the projected attack vector is too small thus failing to induce a decision error. Such a probability is augmented by the non-optimality of the direction of the projected attack vector, since there is no guarantee that such a direction is orthogonal to the decision boundary of the RFS detector. In order to understand when and to which extent the {\em angle mismatch} between the optimal attack direction and the direction of the projected attack vector affects the security of the RSF detector, let us consider the expression of such directions. The direction orthogonal to the decision boundary of the RFS detector is:
\begin{equation}
\label{eq.orth_RFS}
    {\bf e}_{RFS} = \frac{\Sigma_r^{-1} {\bf u}_r}{\| \Sigma_r^{-1} {\bf u}_r \|},
\end{equation}
while the projection of ${\bf e}_{att}$ on the reduced feature space is:
\begin{equation}
\label{eq.eatt_r}
    {\bf e}_{att,r} = \frac{S \Sigma^{-1} {\bf u}}{\| S \Sigma^{-1} {\bf u} \|}.
\end{equation}
In the case of i.i.d. feature, it is easy to see that the above directions coincide given that $\Sigma_r^{-1} = I_{k \times k}$, $\Sigma^{-1} = I_{n \times n}$ and $S S^{T} = I_{k \times k}$.
This is not the case with dependent features. To understand the importance of the angle mismatch in this case, we randomly generated 1000 pairs of covariance matrixes and random selection matrices $S$. Then we plot the histogram of the angle mismatch. The results we got are reported in Fig. \ref{fig.angmis}, for both the cases of non-normalized and normalized features. The figure shows the results for $k = 50, 150$ and $250$ (similar results hold for other values of $k$). Upon inspection of the figure, we can see that the angle mismatch is much larger in the case of normalized features, thus explaining the higher security performance of the randomized detector in this case. For small values of $k$ the mismatch is so large that the attack vector is almost orthogonal to the optimal direction, thus nullifying the effect of the attack. Such an observation also explains why in this case increasing $\alpha$ does not increase the effectiveness of the attack, in fact, the increased value of the magnitude of the attack is {\em wasted} since the direction of the attack is a wrong one.

\begin{figure}[t]
	\centering	
	\subfigure[]{
		\includegraphics[width=0.31\columnwidth, trim={30 180 50 180},clip]{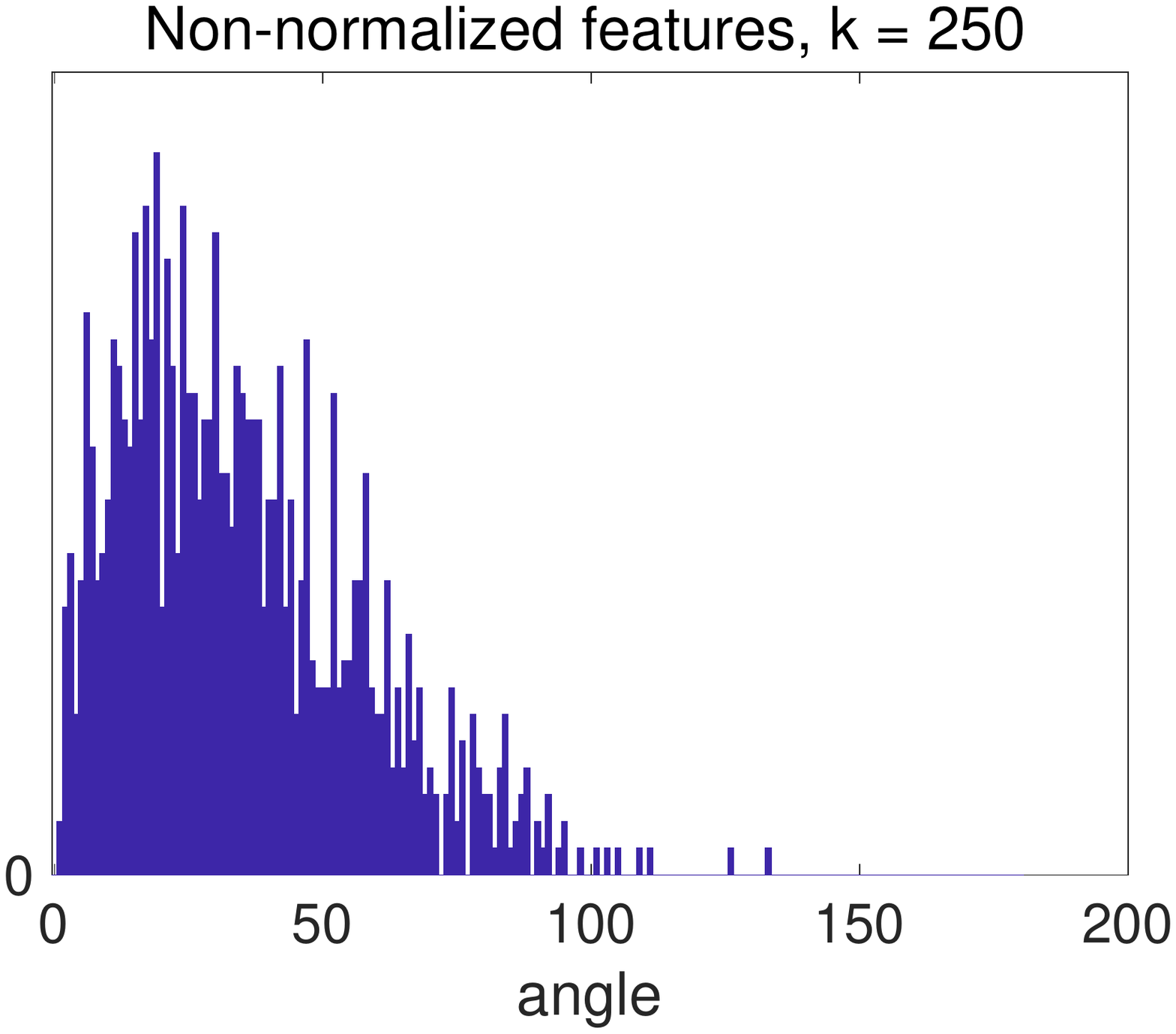}}
	\subfigure[]{
		\includegraphics[width=0.31\columnwidth, trim={30 180 50 180},clip]{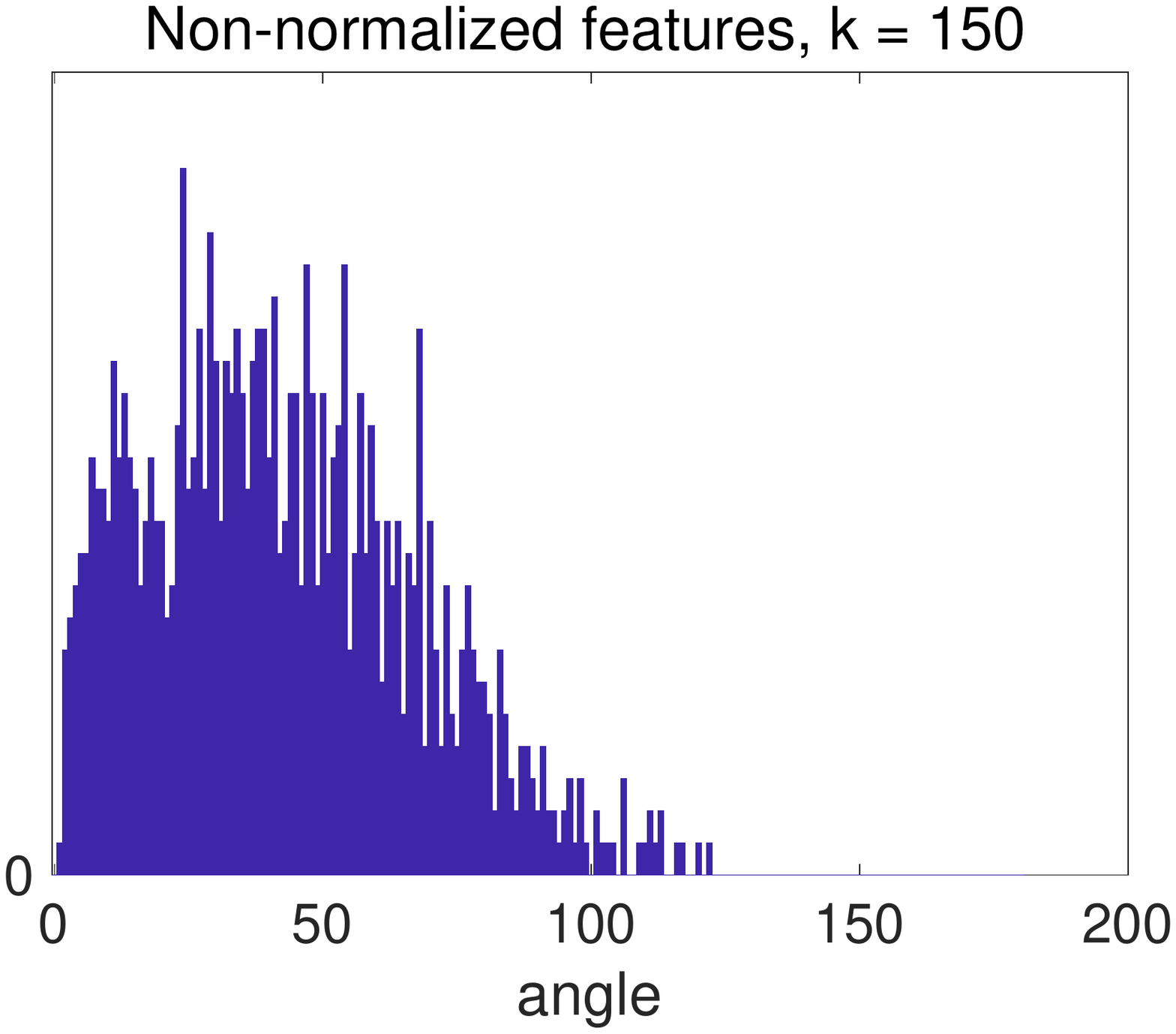}}
			\subfigure[]{
		\includegraphics[width=0.31\columnwidth, trim={40 180 50 180},clip]{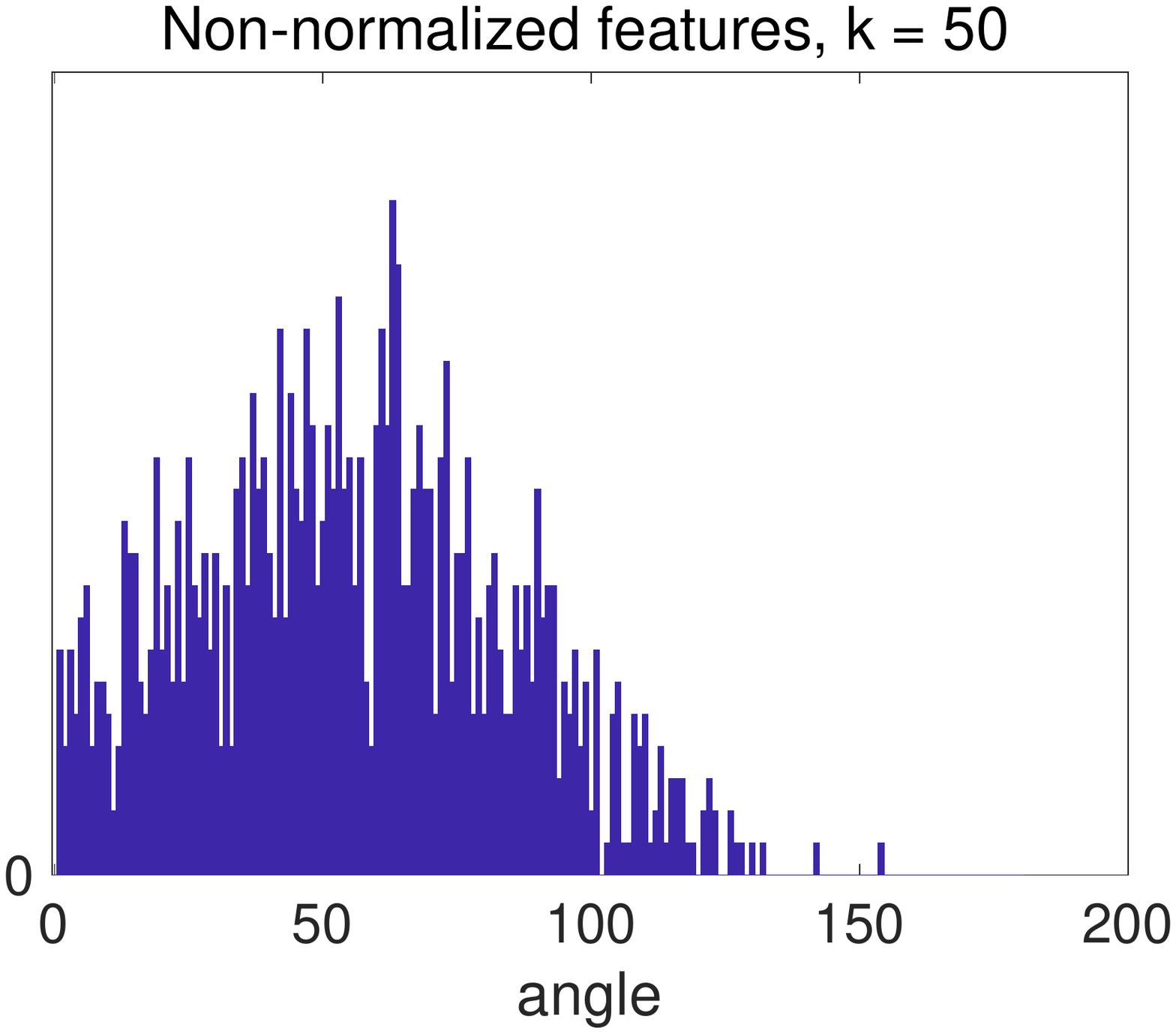}}
	\subfigure[]{
		\includegraphics[width=0.31\columnwidth, trim={40 180 50 180},clip]{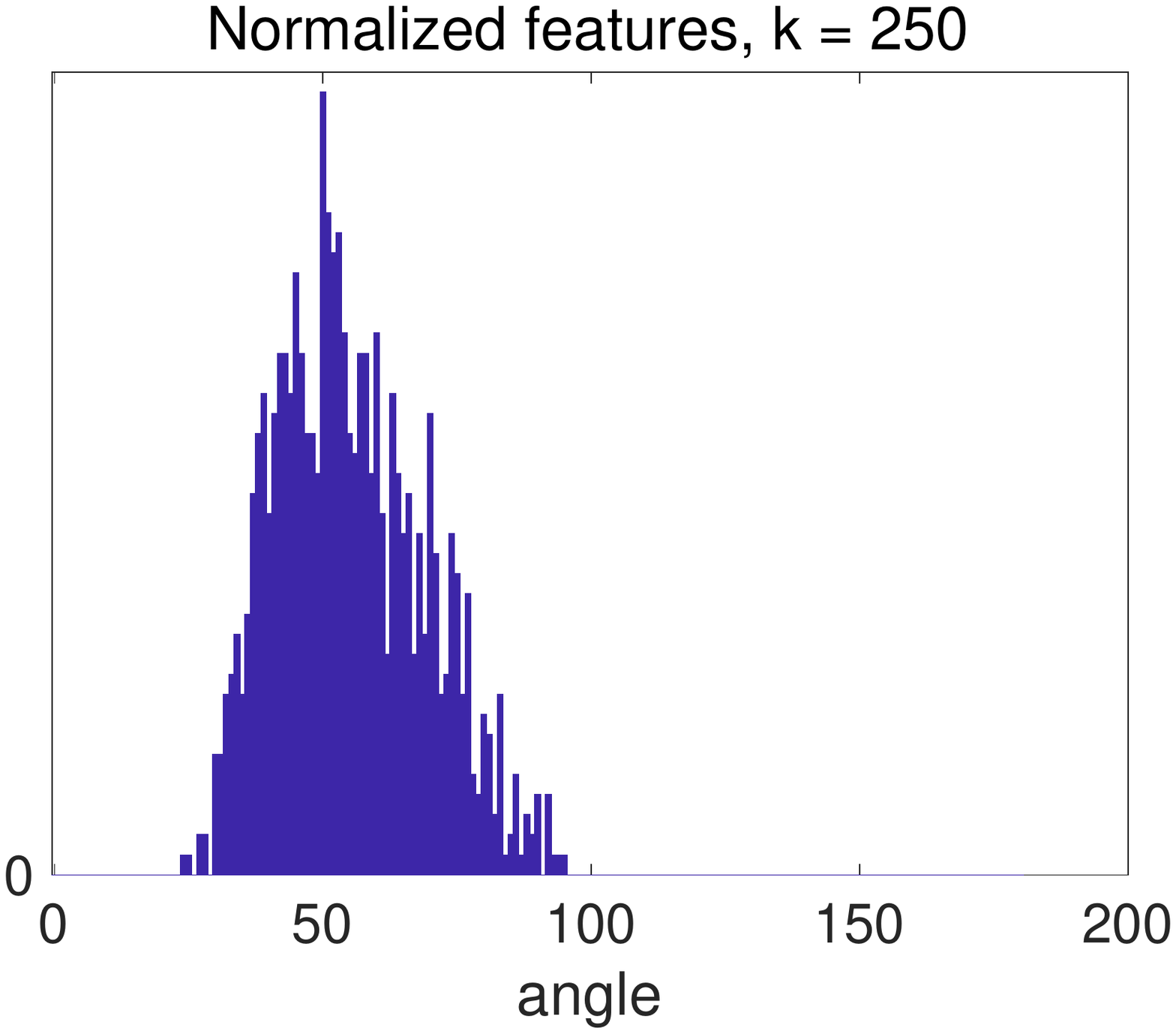}}
					\subfigure[]{
		\includegraphics[width=0.31\columnwidth, trim={40 180 50 180},clip]{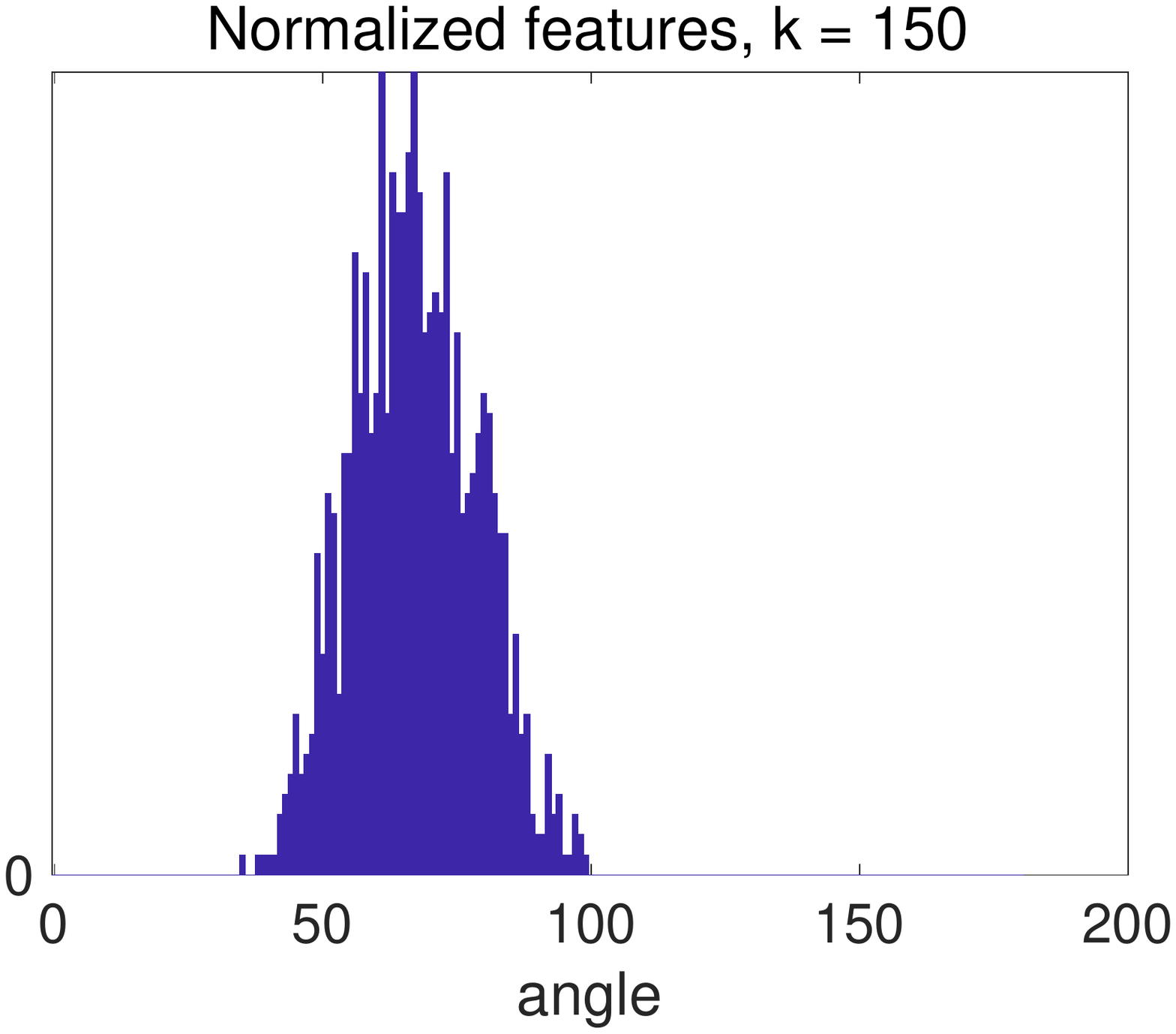}}
	\subfigure[]{
		\includegraphics[width=0.31\columnwidth, trim={40 180 50 180},clip]{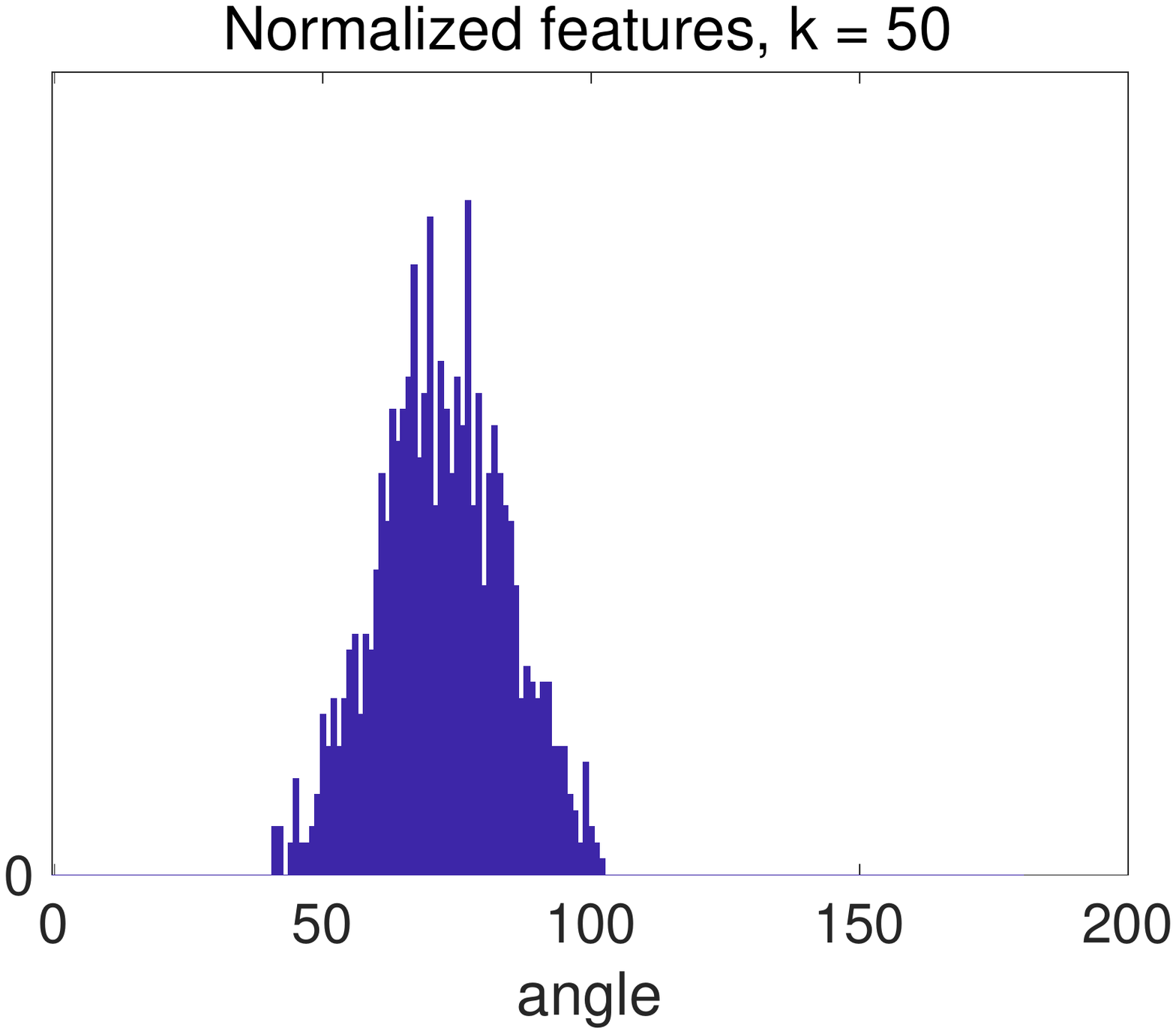}}	
	\caption{Angle mismatch between the projection of the attack vector on the feature space of the RFS detector and the optimal attack direction for the RFS detector. Plots (a)-(c) refer to the case of non-normalized features, while plots (d)-(e) correspond to the case of normalized features.}
\label{fig.angmis} 
\end{figure}

Given that the particular form of the matrix $S$ does not have a significant impact on the performance of the reduced detector, in the following we focus on the RFS case only. An advantage of such an approach is its lower computational complexity. In the RP cases, in fact, the detector must compute the entire feature vector and then choose only $k$ linear combinations. In the RFS case, instead, the detector can compute only the features that it intends to use, avoiding to calculate the non-selected features; this may allow a significant reduction of the complexity, especially for large values of $n$ and small values of $k$.

\section{Application to image manipulation detection}
\label{sec.exa}

The theoretical analysis given in the previous section suggests that a detector based on a randomised subset of features provides a better security with respect to a full-feature detector. The applicability of such an idea to real world applications, however, requires great care, since the assumptions behind the theoretical analysis are ideal ones and are rarely met in practice.  Since the goal of this paper is to improve the security of image forensic techniques against counter-forensic attacks, in this section, we introduce an SVM-based detector based on random feature selection and apply it to two particular image forensic problems, namely the detection of adaptive histogram equalization and the detection of median filtering. The full feature space consists of a subset of SPAM features \cite{Pevny10}, however the SVM detector is trained by relying
only on a random subset of the full feature set. As we will see, the loss of performance of the RFS SVM detector in the absence of attacks is very limited, even for rather small values of $k$.

We also introduce two attacks aiming at deceiving the SVM detector. Both attacks are based on gradient descent, the first one works in the feature domain, while the second operates directly in the pixel domain. The attacks are very powerful since they were able to prevent a correct detection in all the test images. In particular, the attack operating in the pixel domain is a very practical one since it does not require to map back the attack from the feature to the pixel-domain and can prevent a correct detection by introducing a limited distortion in the attacked image. In Section \ref{sec.exp}, we will use these attacks to demonstrate the improved security ensured by the RFS detector.

\subsection{RFS SVM-based detection of image manipulations}
\label{subsec.SVM}

Residual-based features, originally devised for steganalysis applications \cite{Pevny10,Frid12rich}, have been used with success in many image forensic applications, including forgery detection \cite{Cozzolino2015Image}, detection of pixel-domain enhancement, spatial filtering, resampling and lossy compression \cite{Li2016Identification}.
In particular, in this paper, we consider the SPAM feature set \cite{Pevny10}. Since we carried out all our tests on grey-level images, we assume that these features are computed directly on grey-level pixel values, or on the luminance component derived from the RGB colour bands. Feature computation consists of three steps. In the first step, residual values are computed; specifically, the difference 2-D arrays are evaluated along horizontal, vertical and diagonal directions. In the second step, the residual values are truncated so that their  maximum absolute value is equal to $T $. Finally, the co-occurrence matrices  are computed. Depending on the value of $T$ and the order of the co-occurences considered in the computation, different sets of features with different sizes are obtained.
We use second-order SPAM features with $T=3$, for a dimensionality of the feature space of 686.

Based on the SPAM features, we built an SVM detector aiming at revealing Adaptive Histogram Equalization (AHE) and Median Filtering (MF).  With regard to AHE, we considered the contrast-limited algorithm (CLAHE) implemented by Matlab function $ adapthisteq $ with $ cliplimit = 0.02$. Some sample images manipulated in this way are shown in the second column in Fig. \ref{fig:compare}.
With regard to MF, we considered window sizes $3\times3, 5\times5$ and $7\times7,$  (MF3, MF5, MF7).

\begin{figure*}[htbp]
	\centering
	\subfigure[]{
		\label{fig:compare:d} 
		\includegraphics[width=2in]{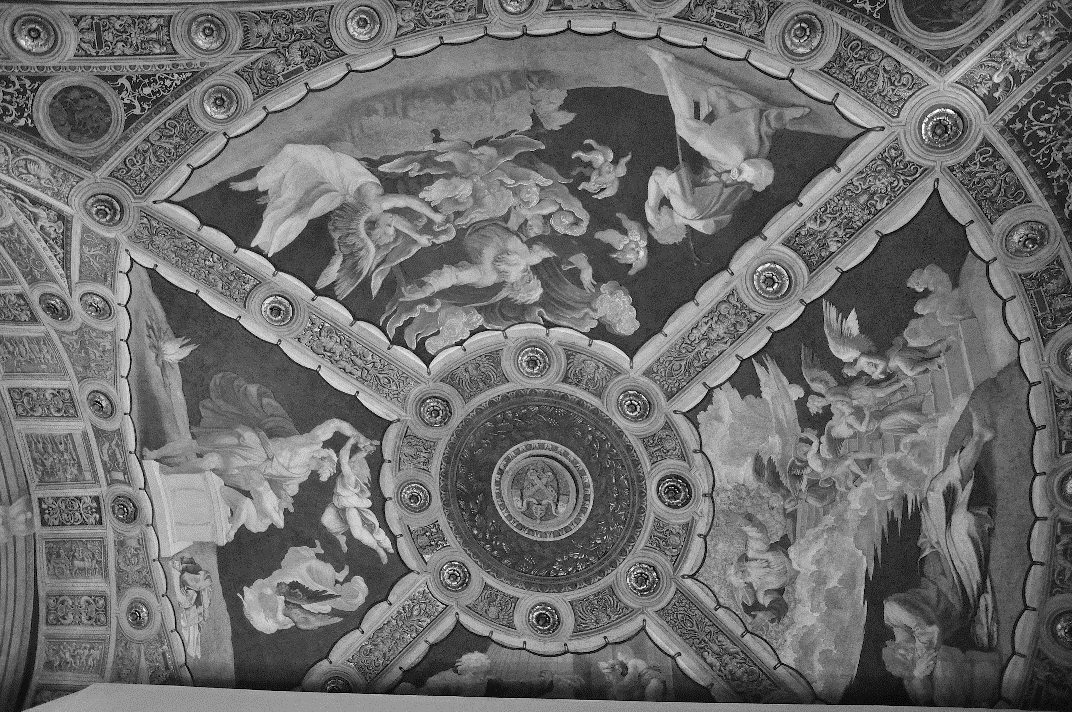}}
	\subfigure[]{
		\label{fig:compare:e} 
		\includegraphics[width=2in]{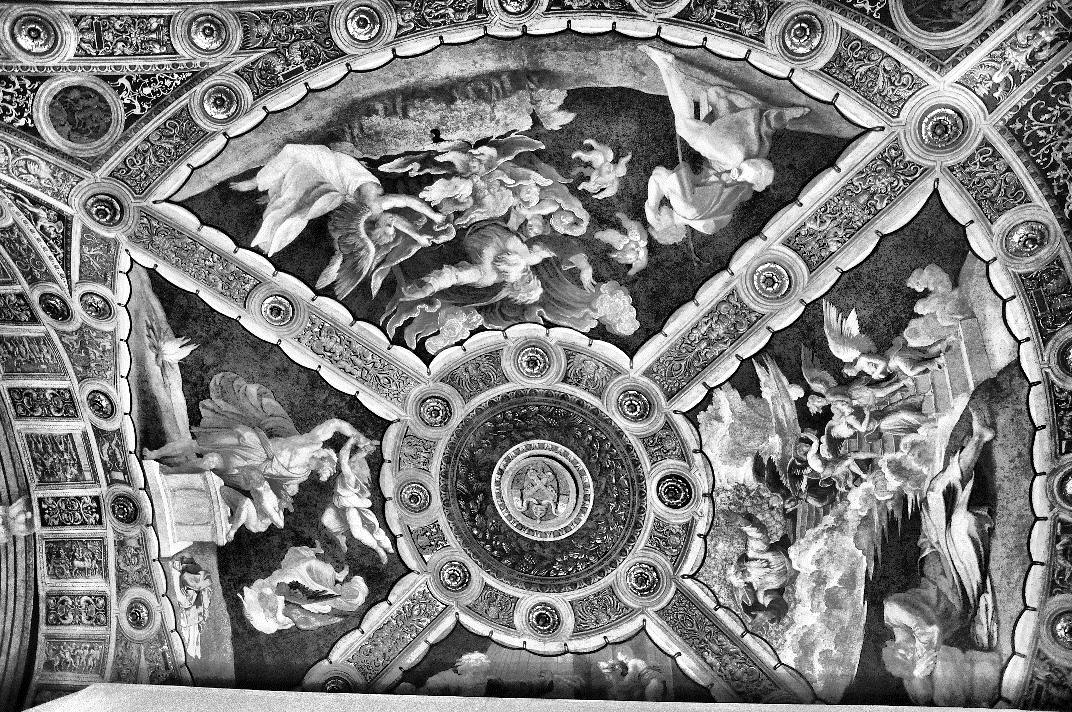}}
	\subfigure[]{
		\label{fig:compare:f} 
		\includegraphics[width=2in]{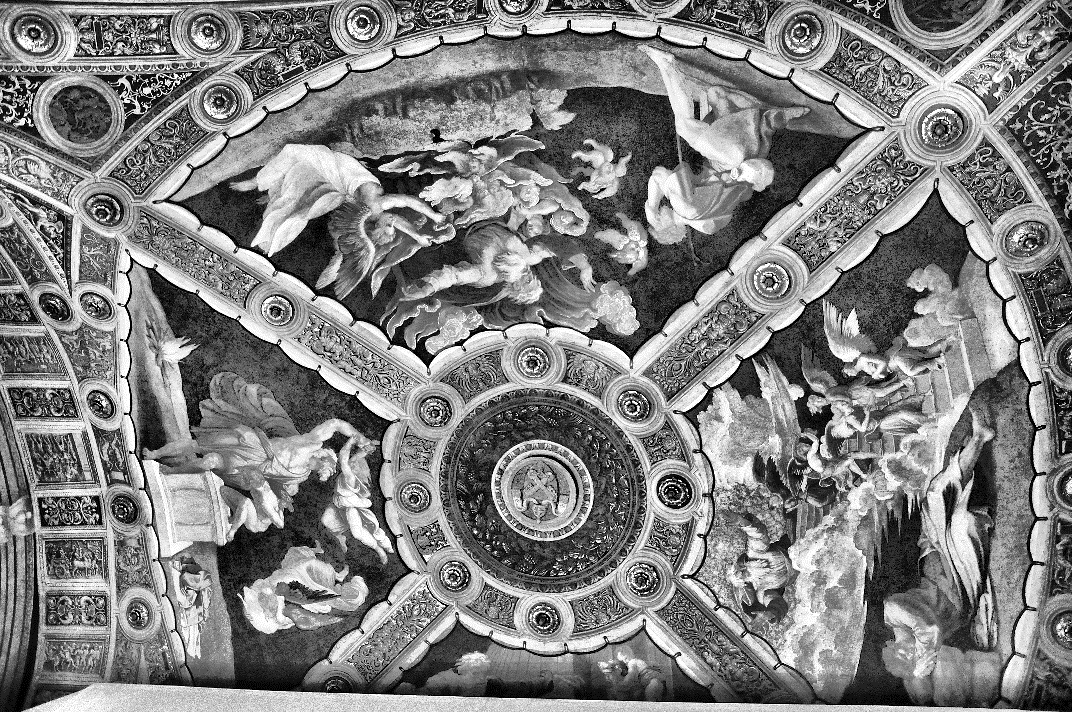}}
	\subfigure[]{
		\label{fig:compare:g} 
		\includegraphics[width=2in]{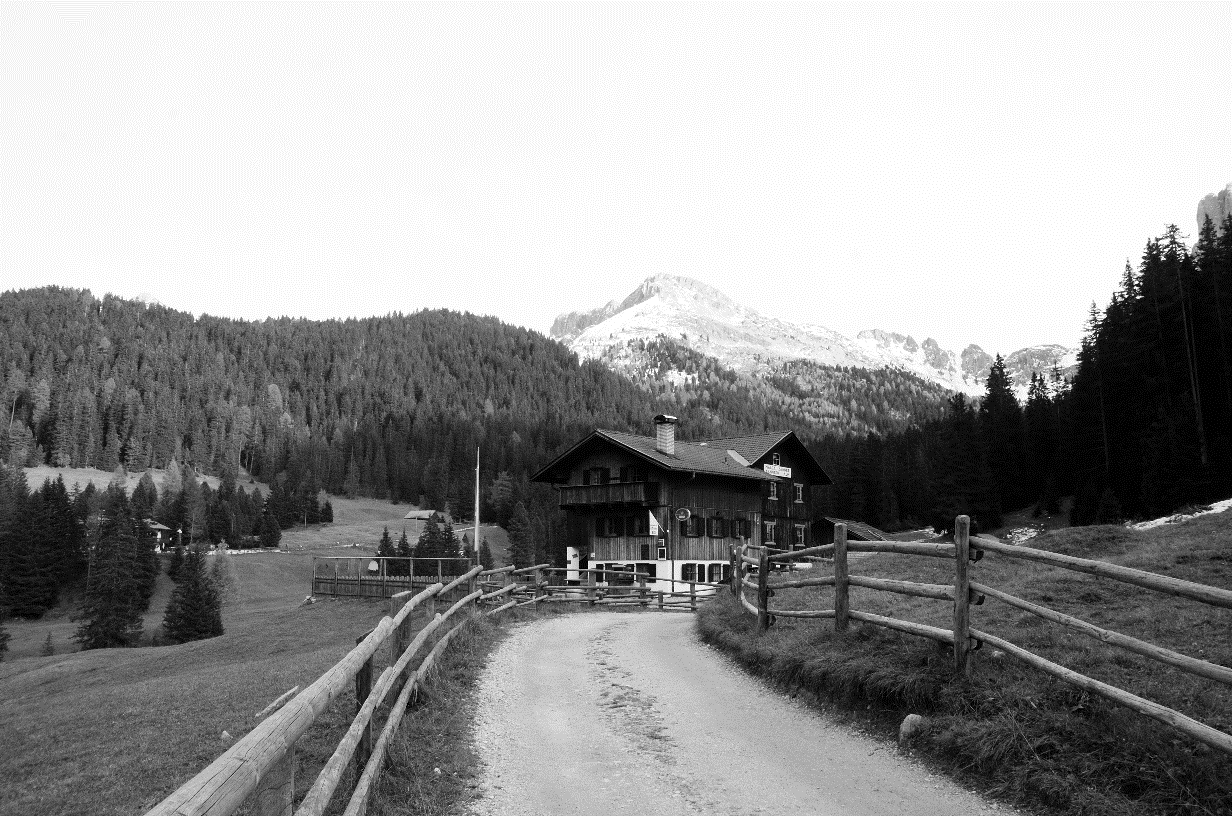}}
	\subfigure[]{
		\label{fig:compare:h} 
		\includegraphics[width=2in]{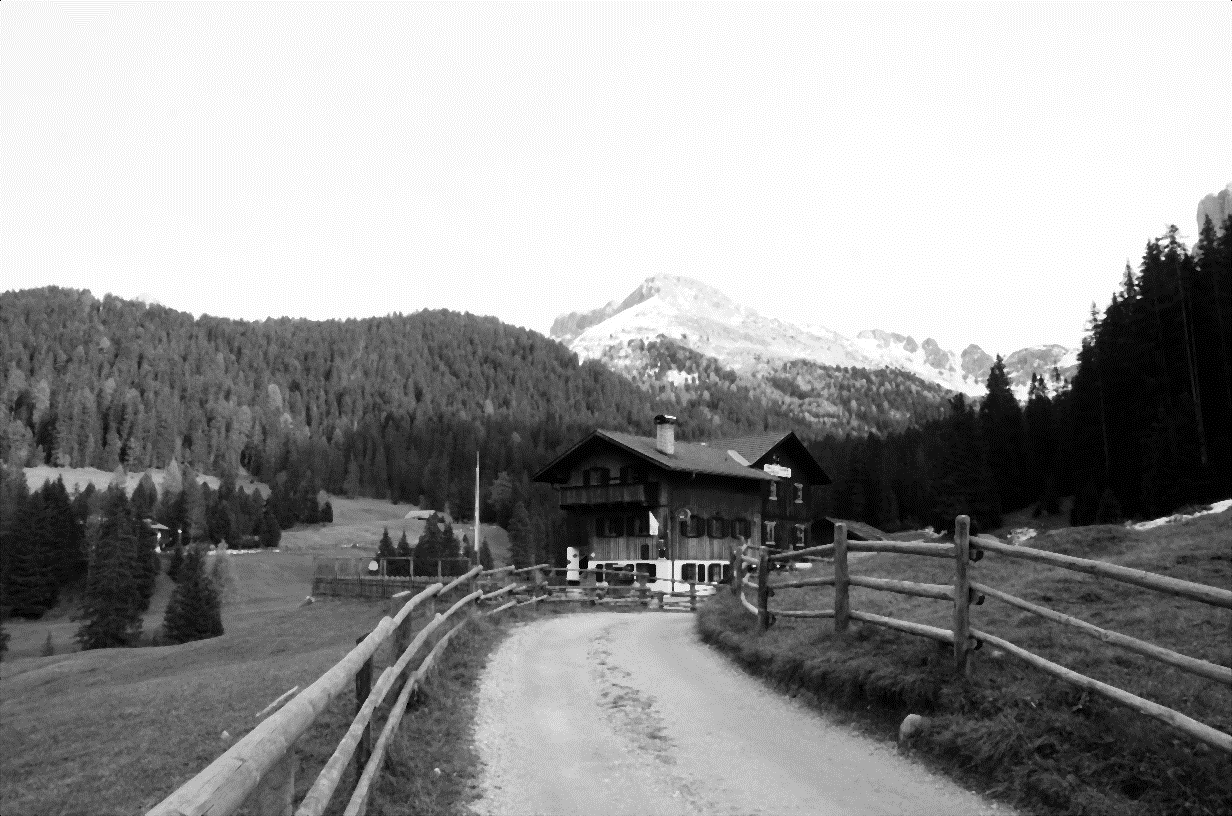}}
	\subfigure[]{
		\label{fig:compare:i} 
		\includegraphics[width=2in]{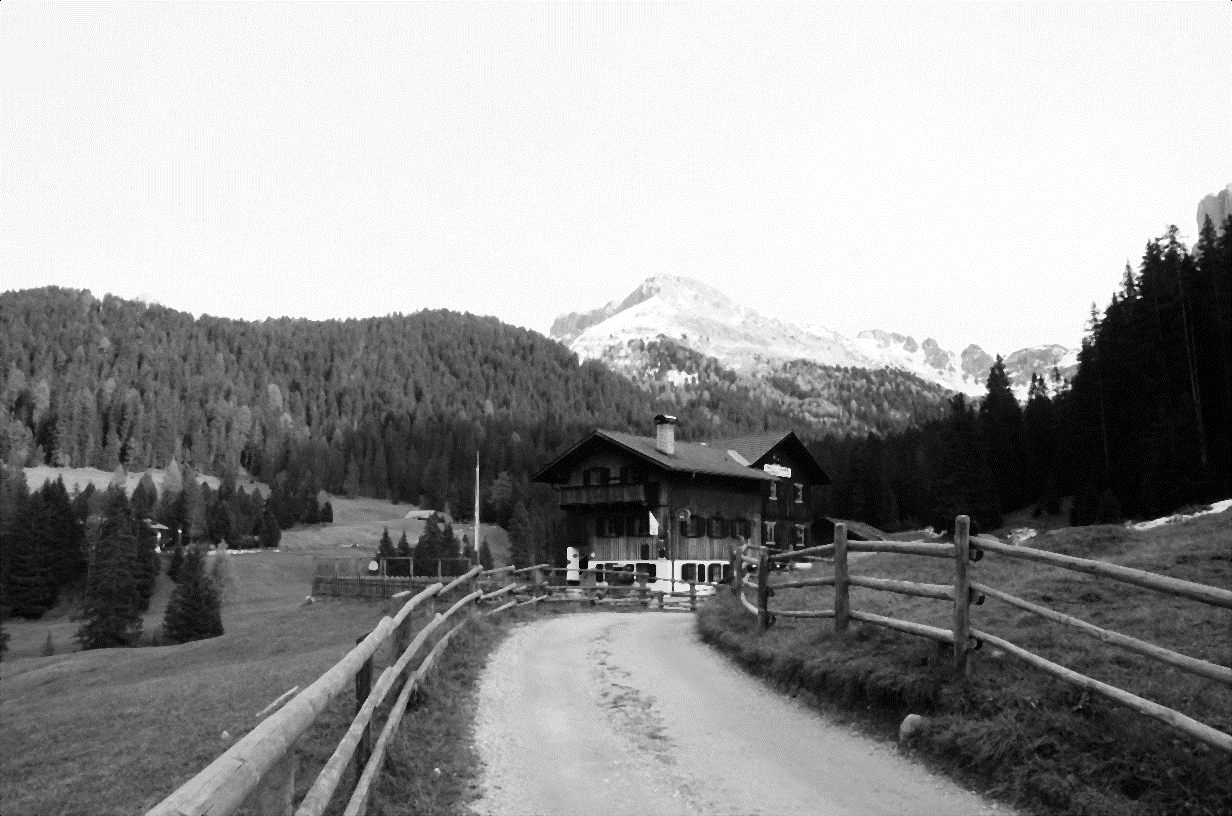}}
	\caption{Examples of manipulated and attacked images: (a) and (d) show the original images;  (b) the image manipulated by AHE; (e) the image manipulated by MF3; (c), and (f) the images after a pixel-domain attack with $\varepsilon = 0.5$.}
	\label{fig:compare} 
\end{figure*}

To train and test the detectors, we used a set of 2000 images from the RAISE-2k dataset \cite{RAISE}. Specifically, we used 1400 images for training and 600 images for testing. To speed up the experiments, the images were downsampled by factor 4 and converted to grayscale. The SVM models are built by using the tools provided by the LibSVM library \cite{chang2011libsvm}. The RBF kernel is used for all the SVMs. The results of the tests are shown in the first row of Table \ref{Tab.errorProb}. All the four detectors got a 100\% accuracy on the test data, thus confirming the excellent capabilities of the SVM trained with the SPAM features to detect global image manipulations like median filtering and histogram equalization.

To double-check that the SVM model does not overfit to the training data, Table \ref{Tab_NumSV} shows the number of support vectors for the SVM detector trained on the full feature set. These numbers are very low compared to the number of examples considered for training, which is 2800 (1400 per class). More specifically, the ratio between the number of support vectors and training examples is around 0.04 - 0.05 (which can be taken as a good indication of the upper bound of the generalisation error \cite{Vapnik}).

\begin{table}[htbp]
	\centering
	\caption{Numbers of support vectors of full-feature SVM model}
	\label{Tab_NumSV}
	\begin{tabular}{ccccc}
		& \textbf{AHE}& \textbf{MF3} & \textbf{MF5} & \textbf{MF7}\\
		\hline
		\textbf{Num. of support vectors} & $62$ & $59$ & $71$  & $69$ \\
		\hline
	\end{tabular}
\end{table}

We also carried out some tests with a linear SVM, however we found that the linear model does not discriminate well the two classes when the number of features decreases to less than 100-200, thus preventing the application of RFS with small values of $k$.

\begin{table}[htbp]
	\centering
	\caption{Error probability of SPAM-based SVM detectors in the absence and presence of attacks ($\varepsilon = 0.5$).}
	\label{Tab.errorProb}
	\begin{tabular}{ccccc}
		& \textbf{AHE}& \textbf{MF3}& \textbf{MF5}& \textbf{MF7}\\
		\hline
		\textbf{Manipulated} & 0\% & 0\%  & 0\%   & 0\%\\
		\textbf{Attack in feature domain} &  100\% & 100\% & 100\% & 100\%\\
		\textbf{Attack in pixel domain} &  100\% & 100\% & 100\%  & 100\%\\		
		\hline
	\end{tabular}
\end{table}

\subsection{Attack in the feature domain}
\label{subsec.att1}

In this section, we describe the feature domain attack we have implemented against the RFS SVM detectors.

The attack in the feature domain has been built by following the system described in \cite{Biggio13}. Specifically, given a feature vector $ {\bf v} $ and a discriminant function $ g(\cdot)$, we assume that the detector decides that ${\bf v} $ belongs to a manipulated image if $g({\bf v} )>0 $, and to an original image otherwise. In this setup, the optimally attacked vector ${\bf v}^*$ is determined by solving the following minimization problem:
\begin{equation}
\label{eq_attFeatDom}
{\bf v}^\star = \mathop{\arg\min}_{{\bf v}':g({\bf v}') \le - \nu} d({\bf v},{\bf v}'),
\end{equation}
where $ d(\cdot,\cdot) $ is a suitable distortion measure and $ \nu $ is a safe margin, which, similarly to the parameter $\alpha$ in Section \ref{sec.theory}, permits to move the attacked vector more or less inside the acceptance region. For an SVM detector, the discriminant function can be written as:
\begin{equation}
g({\bf v}) = \sum\nolimits_i {{\alpha _i}{y_i}k({\bf v},{{\bf v}_i})}  + b,
\end{equation}
where $ \alpha_i $ and $ y_i $ are, respectively, the support value and the label of the $ i-$th support vector $ {\bf v}_i $, and where $k()$ is the kernel function. In our implementation, the minimization problem is solved by using a gradient descent algorithm, where the gradient at each iteration is computed as:
\begin{equation}
\nabla g({\bf v}) = \sum\nolimits_i {{\alpha _i}{y_i}\nabla k({\bf v},{{\bf v}_i})}.
\end{equation}

\begin{table*}[t!]
	\centering
	\caption{Discrimination functions and corresponding gradient for different SVM kernels}
	\label{Tab_grad}
	\begin{tabular}{ccc}
		&\textbf{$g({\bf v})$} & \textbf{$\nabla g({\bf v})$}\\
		\hline
		\textbf{Linear kernel} & $k\left( {{\bf v},{{\bf v}_i}} \right) = {{\bf v}^T}{{\bf v}_i}$ & $\mathop \sum \limits_i {a_i}{y_i}{{\bf v}_i}$ \\
		\textbf{Polynomial kernel} & $k\left( {{\bf v},{{\bf v}_i}} \right) = {\left( {{{\bf v}^T}{{\bf v}_i}}+c \right)^p}$ &  $\mathop \sum \limits_i {a_i}{y_i}p{\left( {{{\bf v}^T}{{\bf v}_i} + c} \right)^{p - 1}}{{\bf v}_i}$\\
		\textbf{RBF kernel}& $k\left( {{\bf v},{{\bf v}_i}} \right) = {\rm{exp}}\left( { - \gamma {{\left| {\left| {{\bf v} - {{\bf v}_i}} \right|} \right|}^2}} \right)$ & $ - \mathop \sum \limits_i 2\gamma {a_i}{y_i}\exp \left( { - \gamma {{\left| {\left| {{\bf v} - {{\bf v}_i}} \right|} \right|}^2}} \right)\left( {{\bf v} - {{\bf v}_i}} \right) $\\
		\hline
	\end{tabular}
\end{table*}

The discrimination function and the corresponding gradient for different kernels are reported in Table \ref{Tab_grad}.

As a matter of fact, in our implementation we used the probabilistic output of the SVM rather than $g(\bf v)$. The probabilistic output is built by mapping $g({\bf v})$ into the [0, 1] range and by letting the value $g({\bf v}) = 0$ correspond to a probabilistic output equal to 0.5. By indicating the probabilistic output of the SVM with $p({\bf v})$, the minimization problem in \eqref{eq_attFeatDom} can be rewritten as:
\begin{equation}
\label{eq_attFeatDom_prob}
{\bf v}^\star = \mathop{\arg\min}_{{\bf v}':p({\bf v}') \le  \varepsilon} d({\bf v},{\bf v}'),
\end{equation}
where $\varepsilon = 0.5$ (equivalent to $\nu = 0$) corresponds to an attacked feature vector lying on the decision boundary and values of $\varepsilon < 0.5$ introduce a safe margin bringing the attacked vector deeper inside the wrong detection region.

We attacked the full-feature SVM detectors for AHE and MF by applying the gradient descent attack in the feature domain; even with $\varepsilon = 0.5$, the attack was able to deceive all the detectors, as reported in the second row of Table \ref{Tab.errorProb}.  The distortion introduced by the attack for $\varepsilon= 0.5 $, $ 0.3 $ and $ 0.1 $ is given in Table. \ref{Tab_ExpSNR}.  As expected the distortion increases for smaller values of $\varepsilon$. We observe that the distortion values reported in the table do not give an immediate indication of the distortion of the attacked image, due to the difficulties of mapping back the attacked feature vector into the pixel domain.


\begin{table}[htbp]
	\centering
	\caption{Average SNR (\textnormal {dB}) computed on the features of 600 attacked images. SNR is defined as the ratio between the energy of the feature vector before the attack and the energy of the distortion. }
	\label{Tab_ExpSNR}
	\begin{tabular}{cccc}
		&\textbf{$\varepsilon = 0.5$} & \textbf{$\varepsilon = 0.3$} & \textbf{$\varepsilon = 0.1$}\\
		\hline
		\textbf{AHE} & $26.01$ & $25.09$ & $23.88$ \\
		\textbf{MF3} & $18.59$ & $17.53$ & $16.1$ \\
		\textbf{MF5} & $19.19$ & $18.15$ & $16.76$ \\
		\textbf{MF7} & $19.15$ & $18.12$ & $16.72$ \\
		\hline
	\end{tabular}
\end{table}

\subsection{Attack in the pixel domain}
\label{subsec.att2}

In a realistic scenario, the attacker will carry out his attack in the pixel domain. Similarly to equation \eqref{eq_attFeatDom_prob}, the goal of the attack in the pixel domain is to solve the following optimisation problem:
\begin{equation}
\label{eq_attPixelDom_prob}
I^\star = \mathop{\arg\min}_{I':p(f(I')) \le  \varepsilon} d(I,I') .
\end{equation}
where $ f(\cdot) $ is the feature extraction function (i.e. $f(I) = \bf{v}$) and, as before, $p$ indicates the probabilistic output of the SVM. Due to the complicated form of $f(\cdot)$ and to the necessity of preserving the integer nature of pixel values, gradient descent can not be applied directly to solve \eqref{eq_attPixelDom_prob}. For this reason we implemented the attack as described in \cite{BarZhip17}, by setting to 20\% the percentage of pixels modified at each iteration of the attack (see \cite{BarZhip17} for more details).

The results of the attack on the same dataset used in Section \ref{subsec.SVM} are reported in the last row of Table \ref{Tab.errorProb}.
Only the results for $\varepsilon= 0.5 $ are shown for simplicity. The results show that the attack always succeeds for all the detectors. The average distortion introduced within the attacked images is reported in Table \ref{Tab_ExpPSNR}, while some examples of attacked images are given in the third column of Fig. \ref{fig:compare}.

\begin{table}[htbp]
	\centering
	\caption{Average PSNR (\textnormal {dB}) of the 600 attacked images}
	\label{Tab_ExpPSNR}
	\begin{tabular}{cccc}
		&\textbf{$\varepsilon = 0.5$} & \textbf{$\varepsilon = 0.3$} & \textbf{$\varepsilon = 0.1$}\\
		\hline
		\textbf{AHE} & $54.49$ & $53.09$ & $51.96$ \\
		\textbf{MF3} & $56.77$ & $56.43$ & $54.17$ \\
		\textbf{MF5} & $56.79$ & $55.51$ & $53.37$ \\
		\textbf{MF7} & $56.43$ & $54.73$ & $52.63$ \\
		\hline
	\end{tabular}
\end{table}

\section{Security of RFS-based manipulation detection against feature-domain and pixel-domain attacks}
\label{sec.exp}

In this section, we show the improved security provided by RFS by testing the adaptive histogram equalization and median filtering detectors introduced in Section \ref{subsec.SVM} against the attacks presented in Sections \ref{subsec.att1} and \ref{subsec.att2}.

\subsection{Experimental methodology}

In our experiments, we used the same setting described in Section \ref{subsec.SVM}. The SVM models were trained by using a randomly selected subset of SPAM features extracted from the training set. We considered two versions of the detectors. The first version relies on SPAM features as they are, with no normalisation. The second versions is built by normalising the SPAM features before feeding them to the SVM. Normalization is applied both during the training and testing phases, and is achieved by dividing each feature by its standard deviation, estimated on the training set. The two versions of the detector correspond to the two scenarios considered in Sect. \ref{subsec.numres} (see Fig. \ref{fig:dep}).
The processed images were attacked by using the two methods presented in Section \ref{subsec.att1} and \ref{subsec.att2}. In both cases, we assumed that the attacker has access to a version of the full-feature SVM detectors trained on the same dataset used by the analyst. The performance of the RFS SVM detectors were then evaluated on both the attacked and non-attacked images for different values of $k$. We repeated the experiments 100 times, each time using a different matrix $S$. Various stopping conditions, namely $\varepsilon = \{0.1, 0.3, 0.5\}$, were considered for both the attacks. For the dimension of the reduced feature set, we considered values of $k$ in the range $\{1,3,5,10,20,50,100,200,300,400,500, 600, 686\}$. The value of $\gamma $ (see Table \ref{Tab_grad}) was obtained by means of 5-fold cross-validation carried out on the training set.
To evaluate the performance in the absence of attacks, we considered both false alarms and missed detections, while for the security in the presence of attacks, we considered only the missed detection probability, since in our setup the goal of the attacker is to induce a missed detection event.

\subsection{Security vs robustness tradeoff: feature domain}

In this section, we describe the performance of the RFS detectors against the feature-domain attack described in Section \ref{subsec.att1}. As we already pointed out, this is an ideal situation for the attacker, since in real applications the attacker usually does not have access to the feature domain and the inverse-mapping  from the feature domain to the pixel domain is often a difficult task. As we will see, despite the setup is more favourable to the attacker than in the case of a pixel domain attack, the RFS detector exhibits a good  security.
\begin{figure*}[htb!]
	\centering	
	\subfigure[]{
		\includegraphics[width=2.2in,height=2.2in]{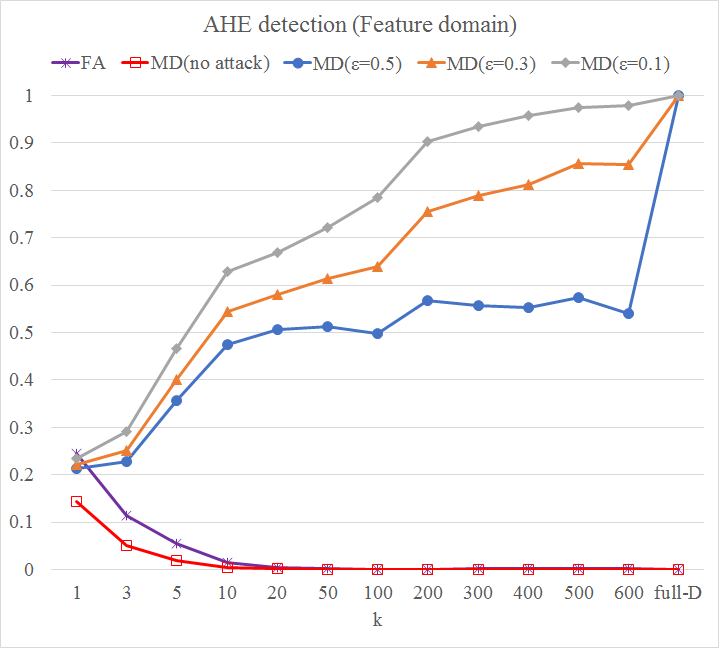}}
	\subfigure[]{
		\includegraphics[width=2.2in,height=2.2in]{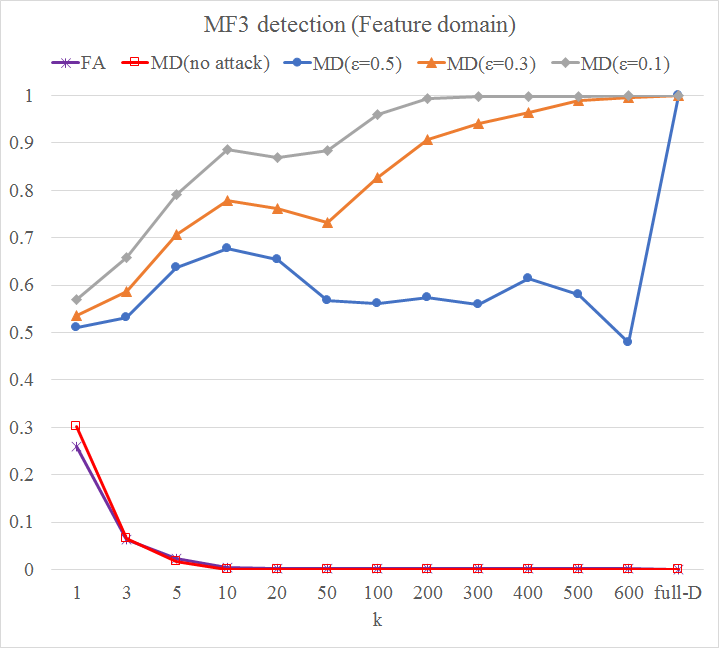}}\\
	\subfigure[]{
		\includegraphics[width=2.2in,height=2.2in]{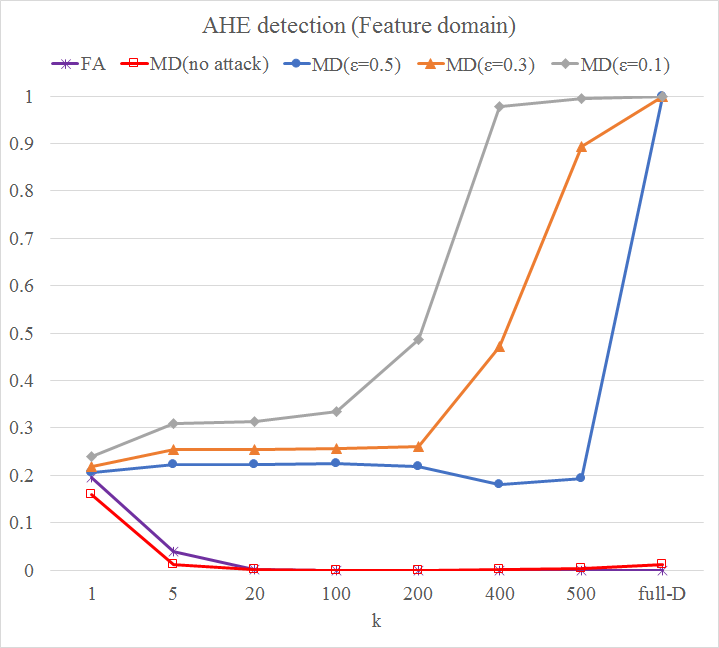}}
	\subfigure[]{
		\includegraphics[width=2.2in,height=2.2in]{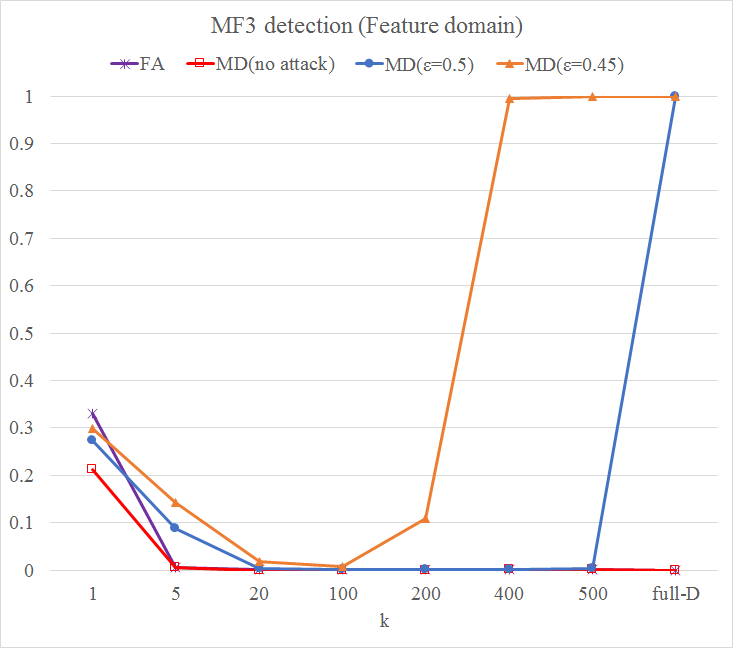}}		
	\caption{Error probability of the randomised feature detector without attacks and under a feature-domain attack, in the case of $ AHE $ and $ MF3 $ detection. (a) and (b) refer to the case of non-normalized features, while  (c) and (d) have been obtained by normalizing the features before feeding them to the detector. The missed detection (MD) probability is reported for both the non-attacked and attacked images, by letting $ \varepsilon = \{0.5; 0.3; 0.1\}$ (in (a), (b) and (c)), and  $ \varepsilon = \{0.5; 0.45\}$  (in (d)). The false alarm (FA) is also reported. The plots have been obtained averaging over 100 random choices of the matrix $S$.}
\label{fig:ExattFeatD} 
\end{figure*}
\begin{figure*}[htb!]
	\centering	
	\subfigure[]{
		\includegraphics[width=2.2in,height=2.2in]{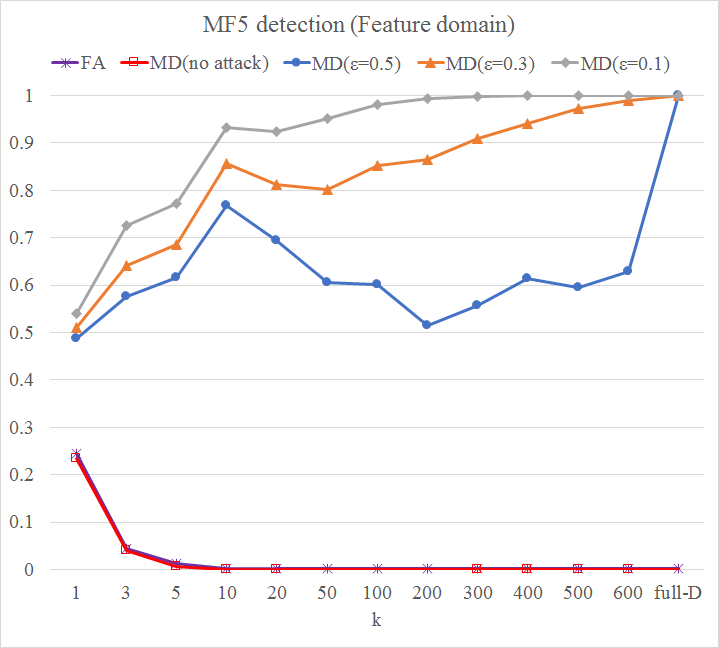}}
	\subfigure[]{
		\includegraphics[width=2.2in,height=2.2in]{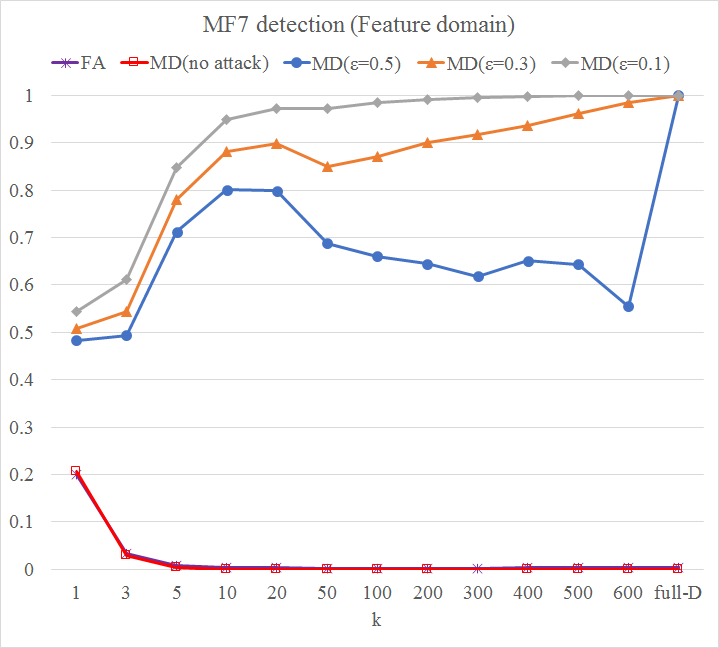}}		
	\caption{Error probability of the randomised feature detector without attacks and under a feature-domain attack, in the case of $ MF5 $ (a) and $ MF7 $ (b) detection with non-normalized features.
The MD probability is reported for both the non-attacked and attacked images, by letting $ \varepsilon = \{0.5; 0.3; 0.1\}$. The FA is also reported. The plots have been obtained averaging over 100 random choices of the matrix $S$. The curves of the FA and the MD for the no-attack case are overlapped.}
	\label{fig:ExattFeatD_MF5-7} 
\end{figure*}
Fig. \ref{fig:ExattFeatD} shows the error probability of the RFS detector as a function of $ k $ with and without attacks for the MF3 and AHE detectors. The plots shown in parts (a) and (b) refer to the case of non-normalized features, while the results shown in parts (c) and (d) have been obtained by normalising the feature vector.
Despite some unavoidable differences, the overall behaviour predicted by the theoretical analysis is confirmed. To start with, the experiments confirm that reducing the dimension of the feature set does not impact much the performance of the detector in the absence of attacks. In fact, even with $k=1$, the probability of correct detection is still around $0.8$ for AHE and $0.7$ for MF3 for the case of non-normalized features (Fig. \ref{fig:ExattFeatD} (a) and (b)). Similar results hold in the normalized feature case (Fig. \ref{fig:ExattFeatD} (c) and (d)). In the presence of attacks, the missed detection probability of the RFS detectors is significantly lower than that of the full-feature detector ($k = 686$ in the figure). For the non-normalized case, we observe that  when the stopping condition of the attack is equal to $0.5$, the missed detection probability drops even for rather large values of $k$, while for $\varepsilon = 0.3$ and $\varepsilon = 0.1$, smaller values of $k$ are needed to make the missed detection probability drop.
With regard to the detector based on normalised features, the experiments confirm the results predicted by the theory, since feature normalization results in a significantly more secure detector. In fact, for the normalized feature case we observe that the missed detection probability immediately drops when $k$ is below a critical very large value of $k$, which depends on the stopping condition of the attack. In particular, for the AHE case, the missed detection error drops below around $k = 400$ for $\varepsilon = 0.1$, $k= 500$ for $\varepsilon = 0.3$ and $k= 600$ for $\varepsilon = 0.5$. For the MF3 case, the situation is even more favorable. First of all we must mention that in this case the stopping conditions were set to $\varepsilon = 0.5$ and $0.45$ only. The reason for such a choice is that reaching $\varepsilon = 0.45$ already requires a large number of iterations and a very strong attack. In fact, the SNR of the attacked feature vector for $\varepsilon = 0.45$ is 4.19 (while it was 7.45 for the AHE case with $\varepsilon = 0.1$).

The results in Fig. \ref{fig:ExattFeatD} clearly show that a suitable value of $k$ can be found where the error probability in the absence of attacks (both false alarm and missed detection) is still negligible and the missed detection probability under attack is significantly smaller than 1.
The results obtained with the MF5 and MF7 detectors are shown in Fig. \ref{fig:ExattFeatD_MF5-7} for the case of non-normalized features. We see that, with respect to the MF3 case, a lower $k$ is necessary to get a small missed detection probability under attacks, especially for the MF7 case.
The motivation for this behaviour is that, in these cases, the separation between the two classes with respect to the intraclass feature variation is very large, thus diminishing the impact of the scale mismatch on the effectiveness of the attack carried out on the full feature set.

\subsection{Security vs robustness tradeoff: pixel domain}
\label{sec.pixdomain}

In this section, we focus on the security of the RFS detectors against the pixel domain attack introduced in \cite{BarZhip17} and briefly described in Section \ref{subsec.att2}. Before going on with the discussion of the results, we stress that attacking the images directly in the pixel domain requires a very high computational effort\footnote{With the attack method in \cite{BarZhip17}, attacking a single image on a 4-core PC, running at 3.3 GHz equipped with 16 GB RAM requires around two hours on the average. By considering that
each curve in the figures shown in this section requires attacking 600 images (and that the effectiveness of the attack must be evaluated for all the images, for all values of $k$ and 100 different selection matrices $S$ for each $k$),
each curve takes more than one month to be computed.}, so we limited our analysis to a lower number of cases, which, however, are sufficient to show the good performance achieved by the RFS detectors.

In Fig. \ref{fig:ExattPixelD}, we show the error probability of the RFS detectors based on non-normalized features. The overall trend of the error probability is similar to that observed for the feature-domain attack, however the error probability is now much smaller, confirming that the possibility of carrying out the attack in the feature domain, would represent a significant advantage for the attacker. By letting $k = 20$, for instance, the missed detection probability in the presence of attacks ranges from 0.36 to 0.55 in the case of AHE and from 0.29 to 0.59 in the case of MF3, while the missed detection probability without attacks is equal to 0.013 and 0, respectively. As for the feature domain attack, the false alarm probability remains low (around 0) even for very small values of $k$, e.g. $k = 10$.

\begin{figure}[t]
	\centering	
	\subfigure[]{
		\includegraphics[width=0.6\columnwidth]{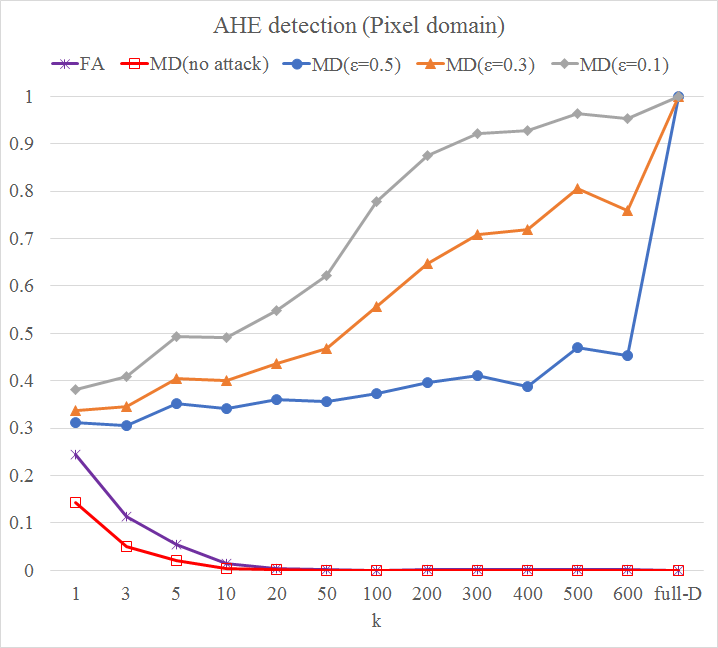}}\\
	\subfigure[]{
		\includegraphics[width=0.6\columnwidth]{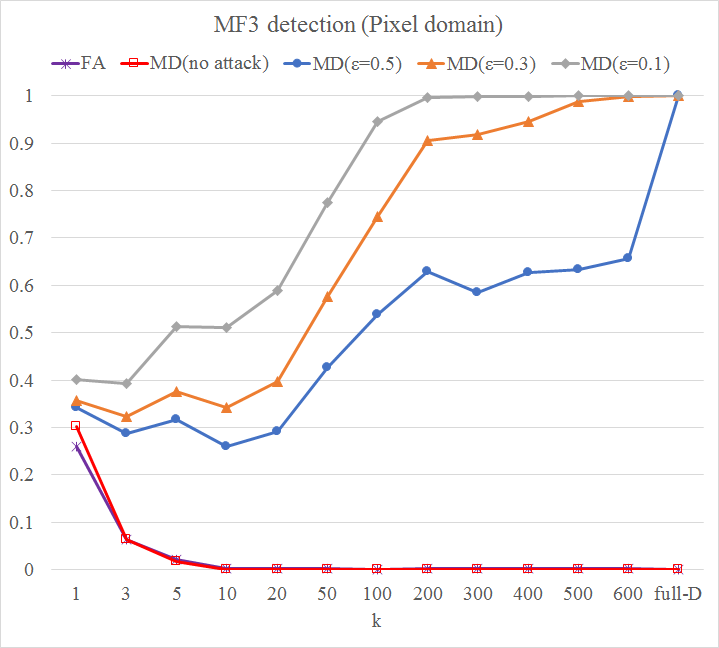}}	
	\caption{Error probability of the randomised feature detector under pixel-domain attack, in the case of $ AHE $ (a) and $ MF3 $ (b) detection with non-normalized features. The plots have been obtained by attacking with  $\varepsilon = \{0.5, 0.3, 0.1\}$ and averaging over 100 random choices of the matrix $S$. The FA and the MD  probabilities in the presence and absence of attacks are reported.}
	\label{fig:ExattPixelD} 
\end{figure}

%

We also carried out some experiments with the AHE-SVM detector relying on normalized features. The results we have got are reported in Fig. \ref{fig:pix_domain_nf}. The general trend of the error probability agrees with that predicted by theory, and confirms that the RFS detector based on normalized features is considerably more secure than the detector based on non-normalized features.

\begin{figure}[h!]
	\centering
		\includegraphics[width=0.6\columnwidth]{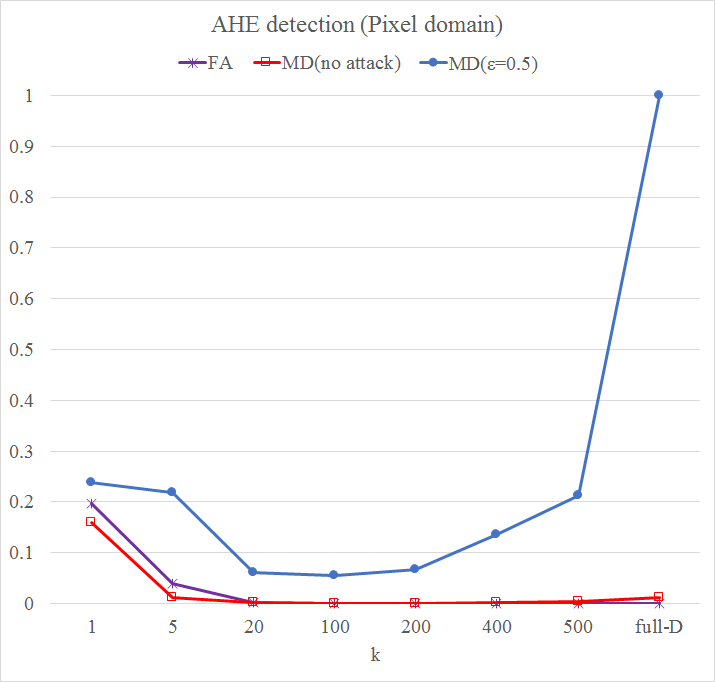}
	\caption{ Error probability of the RFS detector in the  $ AHE $ case with normalized features under a pixel domain attack. The MD probability is reported for both the non-attacked and attacked images, with $\varepsilon = 0.5$. The FA is also reported. The plots have been obtained averaging over 100 random choices of the matrix $S$.}
\label{fig:pix_domain_nf}
\end{figure}

\begin{figure}
	\centering
		\includegraphics[width=\columnwidth]{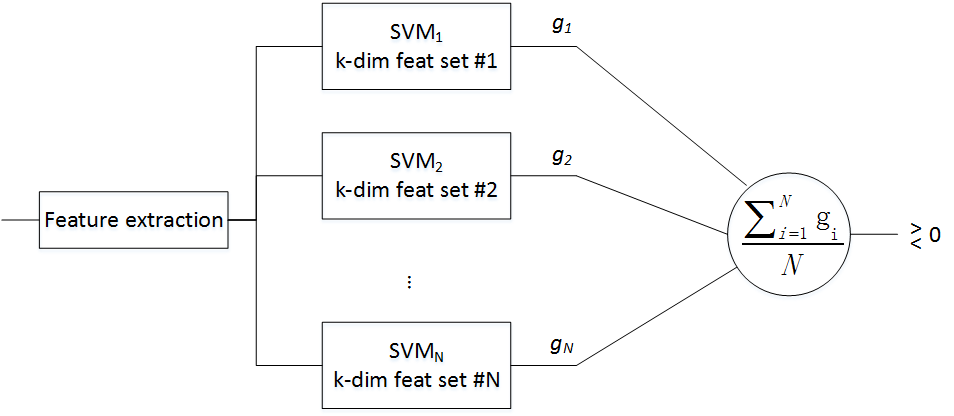}
	\caption{Block diagram of the alternative attack strategy described in Sect. \ref{sec.attackNew}. The number of reduced features $k$ adopted by the real classifier is assumed to be known to the attacker; $g_i$ denotes the discriminant function of the SVM trained on the  $i$-th feature set.}
\label{fig:newattackScheme}
\end{figure}

\subsection{Other attack strategies}
\label{sec.attackNew}
Throughout the paper we have assumed that the attacker designs his attack by targeting the full feature detector. While this is a reasonable approach, if we assume that the attack is aware of the defence mechanism other attack strategies are possible. By assuming that the value of $k$ is publicly available, the attacker could decide to attack a subset of $k$ features  at random,  or he could decide to attack the $k$ most important features chosen based on principal component analysis. While we leave a thorough exploration of alternative attack strategies to a future work, in this section we evaluate the security of the randomized feature detector when the attack targets an {\em expected} version of the classifier, computed by averaging $N$ random predictions obtained by $N$ different classifiers trained on $N$ different random subset of $k$ features. Specifically, we consider the attack setup depicted in Fig. \ref{fig:newattackScheme}. This strategy resembles the Expectation Over Transformation (EOT) attack which has been recently proposed against randomized neural network classifiers \cite{athalye2018obfuscated}. As suggested by its name, such method computes the gradient over the expected transformation of the input. In our case, the transformation consists in the selection of the $k$ features.

For sake of brevity, we implemented the attack depicted in the figure by operating in the feature domain and by considering the case of non-normalized features. With regard to the number of classifiers used to build the attack, we considered the cases $N=50$, and $N = 100$. The results we have got are reported in Fig. \ref{fig:newattack}. Upon inspection of the plots reported in the figure, we see that the effectiveness of the new attack increases with $N$ and it is comparable to that of the attack carried on the full feature detector (indeed always worse than that in the case of median filtering), thus confirming the security improvement achieved by the RFS detector, even in the presence of more elaborated attacks.

\begin{figure*}[htb!]
	\centering	
	\subfigure[]{
		\includegraphics[width=0.65\columnwidth]{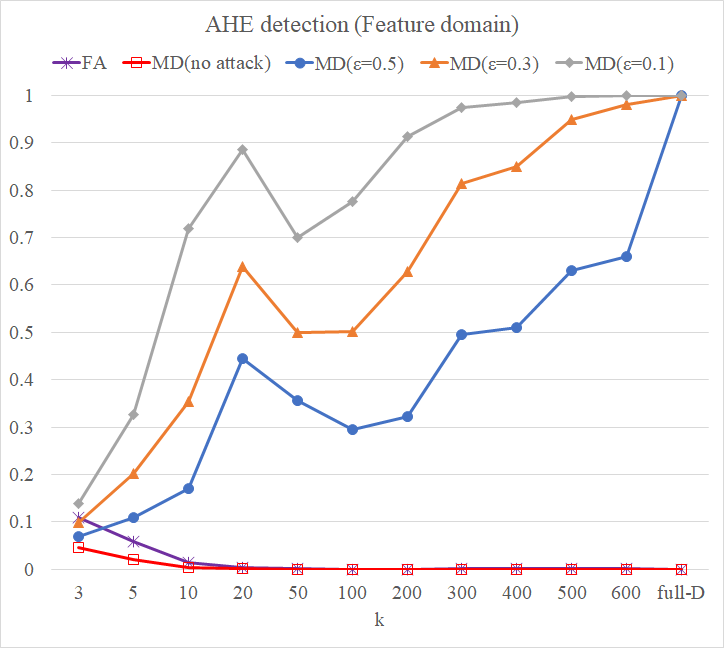}}
	\subfigure[]{
		\includegraphics[width=0.65\columnwidth]{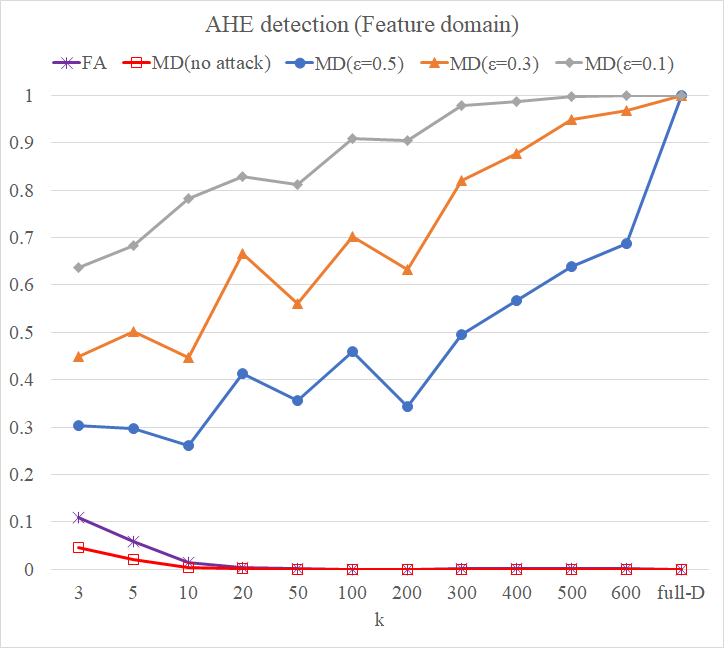}}\\
	\subfigure[]{
		\includegraphics[width=0.65\columnwidth]{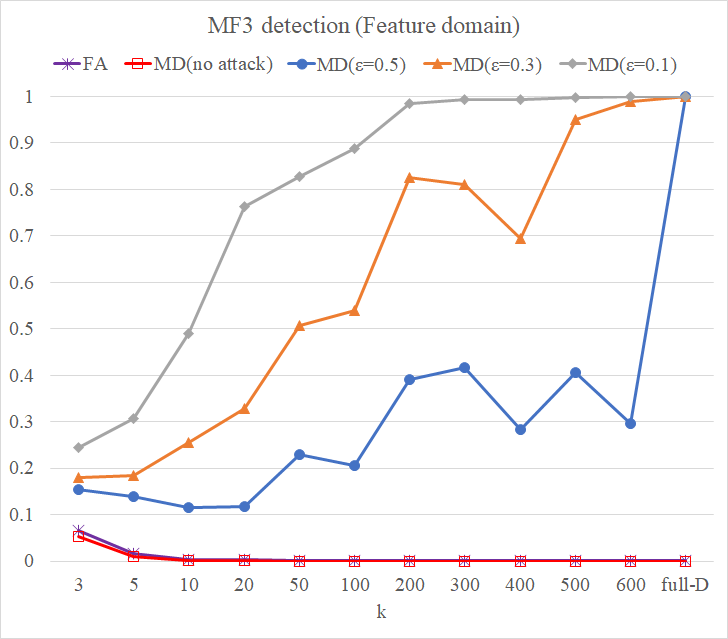}}
	\subfigure[]{
		\includegraphics[width=0.65\columnwidth]{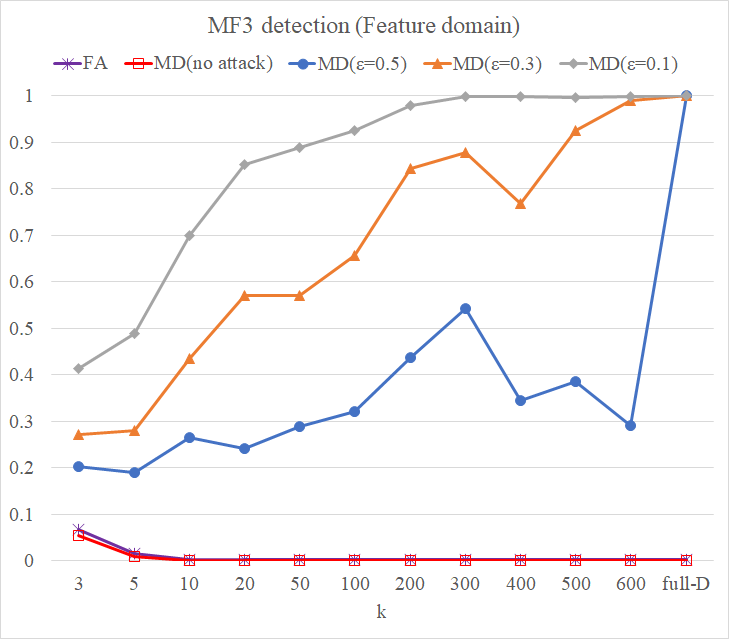}}
	\caption{Error probability of the randomised feature detector without attacks and under a feature-domain attack obtained by averaging 50 (cases (a) and (c)) and 100 (cases (b) and (d)) SVM detectors, in the case of AHE ((a) and (b)) and MF3 ((c) and (d)) with non-normalized features. The MD probability is reported for both the non-attacked and attacked images, by letting $ \varepsilon = \{0.5; 0.3; 0.1\}$. The FA is also reported. The plots have been obtained averaging over 100 random choices of the matrix $S$.}
\label{fig:newattack} 
\end{figure*}

\section{Conclusions}
\label{sec.con}

The use of machine learning tools in the context of image forensics is endangered by the relative ease with which such tools can be deceived by an informed attacker. In this paper, we have considered a particular instantiation of the above problem, wherein the attacker knows the architecture of the detector, the data used to train it and the class of features the detector relies on. Specifically, we introduced a mechanism whereby the analyst designs the detector by randomly choosing the features used by the detector from an initial large set $\VV$ also known by the attacker. We first presented some theoretical results, obtained by relying on a simplified model, suggesting that random feature selection permits to greatly enhance security at the expense of a slight loss of performance in the absence of attacks. Then, we applied our strategy to the detection of two specific image manipulations, namely adaptive histogram equalization and median filtering, by resorting to SVM classification, based on the SPAM feature set. The security analysis, carried out by attacking the two detectors both in the feature and the pixel domain, confirms that security can indeed be improved by means of random feature selection, with a gain that is even more significant than that predicted by the simplified theoretical analysis. In fact, the probability of a successful attack drops from nearly 1 to less than 0.4 in the realistic case that the attack is carried out in the pixel domain. We remark that, while an error probability lower than 0.4 would be preferable, this may be already enough to discourage the attacker in applications wherein ensuring that the attack is successful is a vital requirement for the attacker.

This work is only a first attempt to exploit randomisation, noticeably feature space randomisation, to restore the credibility of the forensic analysis in adversarial settings. From a theoretical perspective, more accurate models could be used to further reduce the gap between the analysis carried out in Section \ref{sec.theory} and the conditions encountered in real applications. From a practical point of view, the use of random feature selection with detectors other than SVMs could be explored, together with the adoption of much larger feature sets, e.g., the entire set of rich features \cite{Frid12rich}. Another interesting research direction consists in the extension of our approach to counter attacks against detectors based on deep learning, specifically convolutional neural networks. In such a case, in fact, the features used by the detector are not chosen by the analyst, since they are determined by the network during the training phase, hence calling for the adoption of other forms of randomisation (see \cite{srivastava2014dropout} and \cite{carlini2017adversarial}  for some preliminary works in this direction).

\section*{Appendix}
\label{sec.Appendix}
\subsection*{Proof of Lemma \ref{eq.lemmastat}.}

\noindent Let $\rho_{r} = \rho_{r,1} - \theta \rho_{r,2}$, where 
\begin{equation}
\label{eq.rho_r1}
\rho_{r,1} = {\bf u}_r^T \Sigma_r^{-1} {\bf v}_r,
\end{equation}
\begin{equation}
\label{eq.rho_r2}
\rho_{r,2} = {\bf w}^T {\bf v}.
\end{equation}
By observing that $\theta$ does not depend on ${\bf v}$ and hence it is a fixed value for a given $S$, the statistics of $\rho_r$ under attack depend only on the statistics of $\rho_{r,1}$ and $\rho_{r,2}$, which are two Gaussian random variables. We now prove that the statistics of  $\rho_{r,1}$ and $\rho_{r,2}$ are given as follows:
\begin{equation}
\label{eq.statrho1}
\begin{aligned}
E[\rho_{r,1}] &= y, & E[\rho_{r,1}^2] &= y + y^2,  &var[\rho_{r,1}] &= y,\\
E[\rho_{r,2}] &= x, & E[\rho_{r,2}^2] &= x + x^2,  &var[\rho_{r,2}] &= x,
\end{aligned}
\end{equation}
and
\begin{equation}
\label{eq.statrho2}
\begin{aligned}
E[\rho_{r,1}\rho_{r,2}] &= xy +y,\\
cov[\rho_{r,1}\rho_{r,2}] &= y,
\end{aligned}
\end{equation}
To derive the above relations, we start by proving that:
\begin{equation}
\label{eq_lemma11}
E[{\bf v}{{\bf v}^T}] = \Sigma  + {\bf u} {{\bf u} ^T}, E[{\bf v}_r{{\bf v}_r^T}] = \Sigma + {\bf u}_r {{\bf u}_r ^T}.
\end{equation}
To prove such properties, we observe that the element in  position $(i,j)$ of matrix $ {\bf v}{\bf v}^T $ is equal to $ v_i v_j $. The element in  position $(i,j)$ of the covariance matrix $\Sigma$ can be computed as:
\begin{equation}
\Sigma = E \left[ {\left( {{v_i} - {u _i}} \right)\left( {{v_j} - {u _j}} \right)} \right] = E\left[ {{v_i}{v_j}} \right] - {u _i}{u _j}.
\end{equation}
Thus, $E\left[ {{v_i}{v_j}} \right] = \Sigma  + {u _i}{u _j}$ and $E[{\bf v}{{\bf v}^T}] = \Sigma  + {\bf u} {{\bf u} ^T}$. In the same way
\begin{equation}
E[{\bf v}_r{{\bf v}_r^T}] = \Sigma_r + {\bf u}_r {{\bf u}_r ^T}.
\end{equation}	
The expectation of $ \rho_{r,1} $ is clearly equal to:
\begin{equation}
E[\rho_{r,1}]  = E[{{\bf u} _r}^T{\Sigma_r^{ - 1}}{{\bf v}_r}] = {{\bf u} _r}^T{\Sigma_r^{ - 1}}{{\bf u}_r} = y,
\end{equation}
hence, based on \eqref{eq_lemma11}, we have:
\begin{equation}
\begin{split}
E[\rho_{r,1}^2] &= E[{{\bf u} _r}^T{\Sigma_r^{ - 1}}{{\bf v}_r}{{\bf v}_r}^T{\Sigma_r^{ - 1}}{{\bf u}_r}] \\
& = {{\bf u} _r}^T{\Sigma_r^{ - 1}} E[{{\bf v}_r}{{\bf v}_r}^T] {\Sigma_r^{ - 1}}{{\bf u}_r} \\
& = {{{\bf u} _r}^T{\Sigma_r^{ - 1}}\left( {\Sigma_r + {{\bf u} _r}{{\bf u} _r}^T} \right){\Sigma_r^{ - 1}}{{\bf u}_r}} \\
& = y+ {y^2}.
\end{split}
\end{equation}
In addition, the variance of $ \rho_{r,1} $ boils down to:
\begin{equation}
var[\rho_{r,1}] = E[\rho_{r,1}^2] -E[\rho_{r,1}]^2 = y.
\end{equation}
In a similar manner, we can prove that:
\begin{equation}
\begin{split}
E[\rho_{r,2}]  &= x, \\
E[\rho_{r,2}^2]  &= x+x^2,  \\
var[\rho_{r,2}] &= x.
\end{split}
\end{equation}
Finally, by exploiting again the properties in \eqref{eq_lemma11}, we can write:
\begin{equation}
\begin{split}
E[\rho_{r,1}\rho_{r,2}]  &= E[{{\bf u} _r}^T{\Sigma_r^{ - 1}}{{\bf v}_r}{\bf v}^T{\Sigma^{ - 1}}{\bf u}] \\
&= E[{{\bf u} _r}^T{\Sigma_r^{ - 1}}{S{\bf v}}{\bf v}^T{\Sigma^{ - 1}}{\bf u}] \\
&= {{\bf u} _r}^T{\Sigma_r^{ - 1}}{S E[{\bf v}}{\bf v}^T] {\Sigma^{ - 1}}{\bf u} \\
&= {{\bf u} _r}^T{\Sigma_r^{ - 1}} S (\Sigma  + {\bf u} {\bf u}^T) {\Sigma^{ - 1}}{\bf u} \\
&=y+xy,
\end{split}
\end{equation}
and then immediately:
\begin{equation}
cov[\rho_{r,1}\rho_{r,2}]  =E[\rho_{r,1}\rho_{r,2}]-E[\rho_{r,1}]E[\rho_{r,2}]=y.
\end{equation}

The lemma follows immediately from relations \eqref{eq.statrho1}-\eqref{eq.statrho2}, by observing that 
\begin{equation}
E[\rho] = E[\rho_{r,1}] - \theta E[\rho_{r,2}]
\end{equation}
and 
\begin{equation}
var[\rho] = var[\rho_1] + \theta^2 var[\rho_2] - 2\theta cov[\rho_{r,1}\rho_{r,2}].
\end{equation}

\section*{Acknowledgment}
This work was supported by the National Key Research and Development of China (No. 
2016YFB0800404), the National Science Foundation of China (Nos. 61532005, 61332012, 61572052,61672090 and U1736213), and the Fundamental Research Funds for the Central Universities (No. 2018JBZ001), and the Key Research and Development of Hebei Province of China (No.17210332). This work has been partially supported by DARPA and AFRL under the research grant number FA8750-16-2-0173. The United States Government is certified to reproduce and distribute reprints for Governmental objectives notwithstanding any copyright notation thereon. The views and conclusions consist of herein are those of the authors and should not be interpreted as necessarily representing the official policies or endorsements, either expressed or implied, of DARPA and AFRL or U.S. Government.

\bibliographystyle{IEEEtran}
\bibliography{RanFeat}

\end{document}